%% file: main.tex
\documentclass[11pt]{article}
\input{preamble}
\usepackage{fullpage}
\usepackage[utf8]{inputenc}

\title{Fast Algorithms for Graph Arboricity and Related Problems}
\author{Ruoxu Cen\\Department of Computer Science\\Duke University
\and Henry Fleischmann\\Computer Science Department\\Carnegie Mellon University
\and George Z. Li\\Computer Science Department\\Carnegie Mellon University
\and Jason Li\\Computer Science Department\\Carnegie Mellon University
\and Debmalya Panigrahi\\Department of Computer Science\\Duke University}
\date{July 2025}

\newcommand{\tO}{\widetilde{O}}
\newcommand{\val}{\mathrm{val}}

\newcommand{\ralert}[1]{\textcolor{red}{#1}}
\newcommand{\cP}{\mathcal{P}}

\newcommand{\str}{\sigma}

\begin{document}

\pagenumbering{gobble}

\maketitle
\begin{abstract}

We give an algorithm for finding the arboricity of a weighted, undirected graph, defined as the minimum number of spanning forests that cover all edges of the graph, in $\sqrt{n} m^{1+o(1)}$ time. This improves on the previous best bound of $\tilde{O}(nm)$ for weighted graphs and $\tilde{O}(m^{3/2})
$ for unweighted graphs (Gabow 1995) for this problem. The running time of our algorithm is dominated by a logarithmic number of calls to a directed global minimum cut subroutine -- if the running time of the latter problem improves to $m^{1+o(1)}$ (thereby matching the running time of maximum flow), the running time of our arboricity algorithm would improve further to $m^{1+o(1)}$.

We also give a new algorithm for computing the entire {\em cut hierarchy} -- laminar multiway cuts with minimum cut ratio in recursively defined induced subgraphs -- in $m n^{1+o(1)}$ time. The cut hierarchy yields the ideal edge loads (Thorup 2001) in a fractional spanning tree packing of the graph which, we show, also corresponds to a max-entropy solution in the spanning tree polytope. For the cut hierarchy problem, the previous best bound was $\tilde{O}(n^2 m)$ for weighted graphs and $\tilde{O}(n m^{3/2})$ for unweighted graphs.

\medskip
\end{abstract}

\clearpage

\pagenumbering{arabic}

\section{Introduction}\label{sec:intro}
\input{intro}

\section{Preliminaries}\label{sec:prelim}
\input{prelim}

\section{Structure of Cut Hierarchy}\label{sec:structure}
\input{structure}

\section{Cut Hierarchy Algorithm}\label{sec:algorithm}
\input{algorithm}

\section{A Faster Algorithm for Arboricity}\label{sec:arboricity-algoritm}

\input{arboricity}

\section{Ideal Tree Packing and Fractional Spanning Tree}\label{sec:entropy}
\input{entropy}

\section{Closing Remarks}\label{sec:closing}
\input{closing}

\section*{Acknowledgments}
RC and DP were supported in part by NSF grants CCF-1955703 and CCF-2329230. Part of this work was done when RC, JL, and DP were visiting the Simons Institute for Theory of Computing, UC Berkeley, as part of the semester program on ``Data Structures and Optimization for Fast Algorithms''. DP also gratefully acknowledges the support of Google Research where he was a visitor during part of this work. GL thanks Jeff Giliberti for mentioning the problem of computing arboricity to him. HF is supported by the National Science Foundation Graduate
Research Fellowship Program under NSF grant DGE2140739. Any opinions, findings,
and conclusions or recommendations expressed in this material are those of the author(s)
and do not necessarily reflect the views of the National Science Foundation.

\bibliographystyle{alpha}
\bibliography{refs}

\appendix

\section{Minimum Directed Cut (Proof of \Cref{lem:alg-sparsify})}\label{sec:partial-sparsify}
\input{sparsification}

\section{Discussion of Trubin's Min-Ratio Cut Algorithm~\cite{Trubin93}}\label{sec:trubin}
\input{trubin}

\end{document}

%% file: preamble.tex
\usepackage{amsmath,amsthm,nicefrac}
\usepackage{amsfonts, amstext}
 \usepackage{libertine}
 \usepackage[libertine]{newtxmath}

\usepackage{typearea}
\typearea{14}
\usepackage[font={small,it}]{caption}

\usepackage{setspace}

\usepackage{enumitem}
\usepackage[breaklinks]{hyperref}
\hypersetup{colorlinks=true,%
            citebordercolor={.6 .6 .6},linkbordercolor={.6 .6 .6},%
citecolor=blue,urlcolor=black,linkcolor=blue,pagecolor=black}

\usepackage[nameinlink]{cleveref}
\Crefname{algocf}{Algorithm}{Algorithms}
\crefname{algocfline}{line}{lines}
\Crefname{invariant}{Invariant}{Invariants}
\Crefname{claim}{Claim}{Claims}
\Crefname{fact}{Fact}{Facts}
\Crefname{subclaim}{Subclaim}{Subclaims}

\usepackage{epsfig}

\usepackage{mathrsfs}
\usepackage{xspace}
\usepackage{soul}
\usepackage{latexsym}
\usepackage{bbm}

\usepackage{svg}
\usepackage{tikz}
\usetikzlibrary{patterns}
\usetikzlibrary{shadows.blur}
\usetikzlibrary{shapes}

\usepackage{framed}

\usepackage[ruled,vlined,linesnumbered,algonl]{algorithm2e}
\usepackage{algpseudocode} 
\SetEndCharOfAlgoLine{}

\SetKwComment{Comment}{\footnotesize$\triangleright$\ }{}

\SetCommentSty{mycommfont}

\def\ShowComment{True} 

\ifdefined\ShowComment

\newcommand{\agnote}[1]{\refstepcounter{note}$\ll${\bf Anupam~\thenote:}
  {\sf \color{blue} #1}$\gg$\marginpar{\tiny\bf AG~\thenote}}
\newcommand{\aknote}[1]{\refstepcounter{note}$\ll${\bf Amit~\thenote:}
  {\sf \color{blue} #1}$\gg$\marginpar{\tiny\bf AK~\thenote}}
\newcommand{\dpnote}[1]{\refstepcounter{note}$\ll${\bf Debmalya~\thenote:}
  {\sf \color{blue} #1}$\gg$\marginpar{\tiny\bf DP~\thenote}}  
  
\newcommand{\redd}[1]{{\color{red}#1}}
\newcommand{\balert}[1]{{\color{blue}#1}}
\newcommand{\alert}[1]{{\color{red}#1}}
\newcommand{\dpalert}[1]{{\color{magenta}#1}}

\else

\newcommand{\agnote}[1]{}
\newcommand{\aknote}[1]{}
\newcommand{\dpnote}[1]{}

\newcommand{\redd}[1]{}
\newcommand{\balert}[1]{}
\newcommand{\alert}[1]{}
\newcommand{\dpalert}[1]{}

\fi

\usepackage{thmtools,thm-restate}

\allowdisplaybreaks

\newcounter{note}[section]
\renewcommand{\thenote}{\thesection.\arabic{note}}

\sethlcolor{fluorescentyellow}

\newcommand{\initOneLiners}{%
    \setlength{\itemsep}{0pt}
    \setlength{\parsep }{0pt}
    \setlength{\topsep }{0pt}
}

\newtheorem{theorem}{Theorem}[section]
\newtheorem{lemma}[theorem]{Lemma}

\newtheorem{fact}[theorem]{Fact}
\newtheorem{corollary}[theorem]{Corollary}

\theoremstyle{definition}
\newtheorem{defn}[theorem]{Definition}

\theoremstyle{remark}
\newtheorem{remark}[theorem]{Remark}

\makeatletter 
\renewcommand{\theinvariant}{(I\@arabic\c@invariant)}
\makeatother

\newcommand{\eps}{\varepsilon}

\newcommand{\poly}{\operatorname{poly}}

\newcommand{\polylog}{\operatorname{polylog}}
\renewcommand{\emptyset}{\varnothing}

\newcommand{\E}{\mathbb{E}}

\newcommand{\nf}{\nicefrac}

\newcommand{\junk}[1]{}
\newcommand{\eat}[1]{}

\newif\ifhideproofs

\ifhideproofs
\usepackage{environ}
\NewEnviron{hide}{}

\fi


%% file: intro.tex
\eat{

\textcolor{red}{For arboricity: Best current algorithm is $\tilde{O}(mn)$ time for weighted graphs and $\tilde{O}(m\sqrt{m})$ for unweighted graphs~\cite{Gabow95}, who devised algorithms based on a parametric version of the directed global minimum cut problem. Earlier work established slower polynomial time algorithms~\cite{DBLP:journals/networks/PicardQ82,DBLP:journals/mp/PadbergW84,DBLP:journals/algorithmica/GabowW92}.
}

\textcolor{red}{
Eppstein~\cite{Eppstein94} obtains a linear-time algorithm for 2-approximation. Plotkin, Shmoys, and Tardos~\cite{PlotkinST95} gives an FPTAS for general packing and covering LP, which works for arboricity. After that, faster FPTAS are obtained by \cite{WorouG16,BlumenstockF20}.
}

\textcolor{red}{
The relation between ideal loads and arboricity is exploited by \cite{deVos25}. In particular, they proved that arboricity = 1 / min ideal load.
}

\textcolor{red}{
The cut hierarchy is implicit in the definition of ideal tree packing from Thorup~\cite{Thorup01,Thorup08}. It is also equivalently defined by Trubin~\cite{Trubin93} (referred to as HS-decomposition; HS stands for homogeneously strong) in the bottom-up way, by repeatedly contracting a skew-densest subgraph.
We will show they are equivalent in \Cref{sec:algorithm}.
Trubin proposed an efficient algorithm for computing cut hierachy. However, we will show in \Cref{sec:trubin} that this algorithm is incorrect. So, the previous best algorithm for cut hierarchy is $O(n)$ calls to Gabow's algorithm for arboricity.
}

}

Consider an undirected graph $G = (V, E)$ with $n$ vertices and $m$ edges with integer edge weights $c_e\in \{1, 2, \ldots, C\}$. The {\em arboricity} of the graph (see, e.g., Schrijver~\cite{Schrijver03}) is defined as the minimum number of disjoint forests that cover all the edges of the graph. (The notion of coverage is that each edge belongs to exactly as many forests as its weight.) Picard and Queyranne~\cite{DBLP:journals/networks/PicardQ82} initiated the algorithmic study of arboricity and showed that the problem is polynomial-time solvable for $O(1)$ edge capacities. This was later extended to general edge capacities by Padberg and Wolsey~\cite{DBLP:journals/mp/PadbergW84}. The running time of these algorithms are large polynomials. A more efficient algorithm was obtained by Gabow and Westermann~\cite{DBLP:journals/algorithmica/GabowW92}, whose running time was $\tO(\min(m^{5/3}, mn))$ when $m = \Omega(n\log n)$, and slightly more for sparser graphs.\footnote{Throughout the paper, we use $\tO(\cdot)$ to hide $\polylog(n)$ factors.} This was further improved in the work of Gabow~\cite{Gabow95}, who leveraged parametric directed minimum cut algorithms to obtain the arboricity of a weighted graph in $\tilde{O}(mn)$ time and that of an unweighted graph in $\tilde{O}(m^{3/2})$ time, respectively. If approximation is allowed, faster algorithms are known.  Eppstein~\cite{Eppstein94} obtained a linear-time algorithm for a 2-approximation of arboricity. In their famous work on positive LPs, Plotkin, Shmoys, and Tardos~\cite{PlotkinST95} gave an FPTAS for finding the arboricity of a graph, whose running time was later improved in \cite{WorouG16,BlumenstockF20}. Nevertheless, for computing the arboricity of a graph exactly, the 30-year-old algorithm of Gabow~\cite{Gabow95} remains the state of the art.

A closely related problem is that of computing the edge loads in an {\em ideal tree packing}, a paradigm introduced by Thorup~\cite{Thorup01,Thorup08}. To define this, recall that a {\em multiway cut} of $G$ (denoted $\cP$) is a vertex partition $V = V_1\sqcup V_2\sqcup \ldots\sqcup V_k$ for any $k\ge 2$, where each $V_i$ is said to be a {\em side} of $\cP$. (If $k=2$, then $\cP$ has only two sides, and is simply a cut of $G$.) The {\em cut ratio} of $\cP$ is defined as $\frac{d(\cP)}{|\cP|-1}$, where $d(\cP) := \sum_{e\in \partial \cP} c_e$ is the sum of weights of all edges with endpoints in two different sides of $\cP$ (call this set of edges $\partial \cP$), and $|\cP|$ is the number of sides of $\cP$. A multiway cut that attains the minimum cut ratio across all the multiway cuts in a graph is called a {\em minimum ratio cut}.
Now, consider the following recursive algorithm. Find a min ratio cut $\cP$ and set $x_e = \frac{|\cP|-1}{d(\cP)}$ for all edges $e\in \partial \cP$, i.e., the reciprocal of the minimum cut ratio. Then, we recurse on the induced subgraphs for each non-singleton component after removing the edges $\partial \cP$. The resulting hierarchy of laminar multiway cuts on induced subgraphs form the {\em cut hierarchy} of the graph and the values $x_e$ on the edges are called the {\em ideal loads}. As observed by de Vos and Christiansen~\cite{deVos25}, the arboricity of a graph is the reciprocal of the minimum ideal load on an edge of the graph. Therefore, we can obtain the arboricity of a graph as a by-product of computing the ideal loads of all edges. 

The cut hierarchy is implicit in the definition of ideal tree packing by Thorup~\cite{Thorup01,Thorup08}. It is also equivalently defined by Trubin~\cite{Trubin93} (referred to as HS-decomposition; HS stands for homogeneously strong). Trubin also claims an efficient algorithm for computing the cut hierarchy, and therefore, the ideal loads. However, we show in \Cref{sec:trubin} that this algorithm is incorrect. Given this observation, the previous best algorithm for computing the cut hierarchy and ideal loads is via $O(n)$ calls to Gabow's algorithm for arboricity~\cite{Gabow95}, which has an overall running time of $\tO(mn^2)$ for weighted graphs and $\tO(m^{3/2}n)$ for unweighted graphs. We also know of a greedy tree packing algorithm to approximate the ideal loads up to small additive error in near-linear time~\cite{Thorup01}.

Before describing our results, we note some other implications of computing the cut hierarchy and the ideal loads. Just as arboricity is the reciprocal of the minimum ideal load, the reciprocal of the {\em maximum} ideal load is called the {\em strength} of the graph (see, e.g., Schrijver~\cite{Schrijver03}). This is also equal to the minimum cut ratio among all the multiway cuts of the graphs, and is denoted
    $\str(G) := \min_{\text{multiway cut }\mathcal{P} \text{ of } G} \frac{d(\mathcal{P})}{|\mathcal{P}|-1}$.
Our cut hierarchy algorithm also returns the strength of the graph, and its running time matches, up to lower order terms, the best previous bound for this problem of $\tO(nm)$. This latter bound is achieved by running a binary search over the values of strength, using an algorithm of Barahona~\cite{Barahona92} that can decide whether the strength is $>1$ in $\tO(nm)$ time. A closely related parameter is the {\em packing number} of a graph, which is the maximum number of spanning trees that can be packed into the graph. By a classical result of Nash-Williams \cite{Nash-Williams61} and Tutte~\cite{Tutte61}, the packing number of a graph is equal to the integral floor of the strength of the graph $\lfloor \sigma\rfloor$. In fact, the entire set of ideal loads yield a {\em fractional spanning tree} of the graph. We show that this fractional tree defined by the ideal loads is precisely the {\em entropy maximizer} in the spanning tree polytope (see \Cref{sec:entropy}), i.e., it is an optimal solution to the convex program $\max_{{\bf x} \in \mathbb{T}} \sum_{e\in E} x_e \ln \nf{1}{x_e}$, where $\mathbb{T}$ is the spanning tree polytope of the graph. 

\eat{

\alert{Ruoxu: We should mention our algorithm matches the best algorithm for computing strength.
Barahona~\cite{Barahona92} can decide whether the strength is $>1$ in $\tO(nm)$ time. By simple scaling and binary search, it can computes the strength in integer weights setting in $\tO(nm\log C)$ time.}

}

\subsection{Our Results}


Our first result is a new algorithm to compute the cut hierarchy, and therefore ideal loads, of a weighted, undirected graph. The algorithm is randomized and its running time is dominated by $\tO(n)$ calls\footnote{In our formal theorems, we make the dependence on $C$ explicit, but in our informal descriptions, we assume that $C$ is polynomially bounded in $n$ for brevity.}
to any maximum flow algorithm on {\em directed} graphs with $O(m)$ edges and $O(m)$ vertices.  Using the current state of the art max-flow algorithms~\cite{chen2023almost,van2023deterministic}, this results in a running time of $nm^{1+o(1)}$. This improves on the previous best algorithm that calls Gabow's arboricity algorithm~\cite{Gabow95} $O(n)$ times, and therefore has a running time of $\tO(m n^2)$ for weighted graphs and $\tO(m^{3/2}n)$ for unweighted graphs. 

\begin{restatable}{theorem}{cuthierarchyAlg}\label{thm:cut-hierarchy}
There is a randomized algorithm that, given an undirected graph $G$ on $m$ edges and $n$ vertices with integer edge weights in $\{1, 2, \ldots, C\}$, outputs the cut hierarchy (and hence ideal edge loads) of $G$ in $\tO(nm\log C)$ time plus $\tO(n\log C)$ calls to a maximum flow subroutine on directed graphs with $O(m)$ edges and $O(m)$ vertices and integer edge weights at most $\poly(n, C)$.
\end{restatable}

Recall that arboricity is given by the reciprocal of the minimum ideal load. Therefore, the above theorem yields an algorithm for computing arboricity in the same running time of $mn^{1+o(1)}$, which is slightly slower than the previous best bounds for this problem. 
We improve this bound further by giving a randomized algorithm for arboricity whose running time is dominated by $\tO(1)$ calls to two subroutines: any maximum flow algorithm on {\em directed} graphs with $O(m)$ edges and $O(m)$ vertices, and any global minimum cut algorithm on {\em directed} graphs with $O(m)$ edges and $O(n)$ vertices.  Using the current state of the art maximum flow algorithms~\cite{chen2023almost,van2023deterministic} and global minimum cut algorithms in directed graphs~\cite{CLNPQS21}, this results in a running time of $\sqrt{n} m^{1+o(1)}$. This improves on the previous best algorithm for arboricity~\cite{Gabow95} that has a running time of $\tO(mn)$ for weighted graphs and $\tO(m^{3/2})$ for unweighted graphs.
Moreover, if the running time of directed minimum cut improves to almost linear time, it will immediately improve the running time of the arboricity algorithm to almost linear time as well.

\begin{restatable}{theorem}{arboricityAlg}\label{thm:arboricity-intro}
    There is a randomized algorithm that, given an undirected graph $G$ on $m$ edges and $n$ vertices with integer edge weights in $\{1, 2, \ldots, C\}$, outputs the arboricity of $G$ in $O(m\log (nC))$ time plus $O(\log (nC))$ calls to a maximum flow subroutine on directed graphs with $O(m)$ edges and $O(m)$ vertices and $O(\log (nC))$ calls to a global minimum cut subroutine on directed graphs with $O(m)$ edges and $O(n)$ vertices, both with integer edge weights at most $\poly(n, C)$.
\end{restatable}

\eat{

\alert{Please check the terms in the above theorem, whether it is randomized or deterministic, and use restatable to make this theorem identical to \Cref{thm:arboricity}.}

}


\subsection{Our Techniques}
Recall that in defining the cut hierarchy, we used a top down recursive definition where the algorithm repeatedly computes a minimum ratio cut on an induced subgraph. In designing our cut hierarchy algorithm, however, we take a bottom up approach. We start with a multiway cut on an induced subgraph at the bottom layer of the hierarchy, and end with the minimum ratio cut that appears at the top level of the hierarchy. Note that any multiway cut at the bottom level of the cut hierarchy must comprise singleton vertices as its sides, since the leaves of the cut hierarchy are the vertices of the graph. (We call this a {\em star set}.) Suppose we had an algorithm to identify a star set $S$. Then, we can use this algorithm to reveal the entire cut hierarchy as follows: Contract the star set $S$ returned by the algorithm and recursively find the cut hierarchy of the contracted graph. In this recursively constructed cut hierarchy, $S$ appears as a leaf since it is a vertex of the contracted graph. Now, return the cut hierarchy for the original graph where we replace leaf $S$ with a star centered at $S$, whose leaves are the vertices in $S$.

But, how do we identify a star set? Intuitively, a minimum ratio cut splits the graph into induced components that are denser than the overall graph. Repeating this argument along each branch of the cut hierarchy, we should expect that the star sets are denser than the induced subgraphs represented by their ancestors, and since they are at the bottom of the hierarchy, there are no denser subgraphs contained in them. To make this formal, we introduce the notion of a {\em dense core}, which is an induced subgraph of a graph that is denser than induced subgraphs containing them, as well as induced subgraphs contained in them. Our main structural lemma connect these two concepts: we show that {\em every dense core in a graph is necessarily a star set of the cut hierarchy}. 

We are left with the problem of identifying a dense core. A natural idea is to find the induced subgraph of maximum (rather than maximal) edge density. (Our notion of edge density is slightly different from the usual one, but this is a finer point that we ignore in this summary.) For this purpose, we use variants of previously known constructions due to Goldberg~\cite{Goldberg84} and Gabow~\cite{Gabow95} that define analogous minimum cut problems in directed graphs to identify the densest subgraph. We could use an off-the-shelf directed min-cut algorithm to complete our algorithm, and this would suffice if there were close-to-linear time bounds for the directed min-cut problem. Unfortunately, the current fastest directed min-cut algorithm~\cite{CLNPQS21} is slower---it runs in $m^{1+o(1)}\sqrt{n}$ time. This still suffices if the resulting star set has $\Omega(\sqrt{n})$ vertices, since the running time of directed min-cut can be amortized over the size of the vertex set that we contract. But, what if the star set is only of size $o(\sqrt{n})$? To handle this possibility, we redesign the directed min-cut algorithm to run in $m^{1+o(1)} k$ time, where $k$ is the size of the star set that we identify and contract. This redesign requires a careful binary search for the correct value of $k$, where for any guessed $k$, we give an $m^{1+o(1)} k$-time verification algorithm that either confirms that a vertex set is indeed a dense core if the guess is at least the true value of $k$, or identifies that the guess of $k$ is too small. 

\eat{

\alert{Add a sketch of the optimization over the above algorithm for arboricity. Start with what you would get if you simply used the above algorithm, identify why this is suboptimal conceptually, and then describe how you exploit this suboptimality to obtain a faster algorithm for arboricity.}

}

Note that the cut hierarchy also yields the ideal loads for all the edges of the graph. This is because for any edge, its ideal load is simply the reciprocal of the cut ratio of the unique multiway cut containing the edge in the cut hierarchy. Once we have the ideal loads, we can return the reciprocal of the smallest ideal load, equivalently the maximum cut ratio among all multiway cuts in the cut hierarchy, as the arboricity of the graph. But, this takes $nm^{1+o(1)}$ time. Can we do better? Note that arboricity is exactly the maximum over cut ratios of star nodes in the cut hierarchy. Since we are building the cut hierarchy in a bottom-up manner, it is unnecessary to build the whole hierarchy for the purpose of computing arboricity. As alluded to above, we can identify the densest subgraph by calling a global minimum cut subroutine in an appropriately defined directed graph. We make this connection precise by giving a direct reduction from arboricity to the directed minimum cut problem--we show that it suffices to make $\polylog(n)$ directed minimum cut calls, plus $\polylog(n)$ maximum flow calls and $\tO(m)$ time outside these calls, in order to compute arboricity. This yields an algorithm for computing arboricity in $\sqrt{n} m^{1+o(1)}$ time, the dominant cost being the running time of the directed minimum cut subroutine~\cite{CLNPQS21}.

\smallskip\noindent{\bf Roadmap.}
The key structural property relating star sets and dense cores appears in \Cref{sec:structure}. This is then used in our algorithm for the cut hierarchy and ideal loads, which is given in \Cref{sec:algorithm}. The algorithm for arboricity is given in \Cref{sec:arboricity-algoritm}. Finally, in \Cref{sec:entropy}, we relate the ideal loads to the entropy-maximizing fractional spanning tree. We use an adaptation of the global directed min cut algorithm~\cite{CLNPQS21} as a subroutine in our algorithm; the proof appears in \Cref{sec:partial-sparsify} for completeness. \Cref{sec:trubin} describes the algorithm of Trubin~\cite{Trubin93} for cut hierarchy and shows that it is incorrect.


%% file: prelim.tex
\subsection{Basic Notations for Graphs and Cuts}
A multiway cut of an undirected graph $G = (V, E)$ is a partition of $V$, i.e., a disjoint collection of subsets of $V$ such that their union is $V$. Each subset in this partition is called a side of the multiway cut. 
For a multiway cut $\mathcal{P}$, we use $\partial \mathcal{P}$ to denote the set of cut edges, i.e., edges whose two endpoints are in different sides of $\mathcal{P}$.
Sometimes, we abuse the term multiway cut to refer to the edge set of a multiway cut.
The all-singleton cut refers to the multiway cut that partitions $V$ into singleton vertices, i.e., $\{\{v\}: v\in V\}$.
A multiway cut with two sides $(S, V\setminus S)$ is a cut, and its edge set is denoted $\partial S$. 

We denote edge weights $c:E\to \mathbb{R}^{\ge 0}$. We extend $c$ to a set function $2^E\to \mathbb{R}^{\ge 0}$ by taking the sum. So, $c(\partial \cP)=\sum_{e\in \partial \cP} c(e)$ is the cut value of $\cP$, and $c(E)$ is the total edge weight.
We abbreviate $c(\partial \mathcal{P})$ by $d(\mathcal{P})$; for a cut $S$, we abbreviate $c(\partial S)$ by $d(S)$. We use a subscript to disambiguate functions defined on different graphs.

We use $G[S]$ to denote the induced subgraph on a vertex set $S$, and use $E[S]$ to denote the edge set of $G[S]$. Contracting a vertex set $S$ means to replace $S$ by a new node $s$, replace all edges $(u, v), u\in S, v\notin S$ by $(s, v)$, and delete all edges in $E[S]$. We use $G/S$ to denote the graph formed by contracting $S$ from $G$. For $U\supseteq S$ or $U\cap S=\emptyset$, we use $U/S$ to denote the result of $U$ after contracting $S$, i.e., $U/S=(U\setminus S)\cup\{s\}$ if $U\supseteq S$ and $U/S=U$ if $U\cap S=\emptyset$. For a partition $\mathcal{P}$ of $V$ and a set $S$ contained in a side of $\mathcal{P}$, we use $\mathcal{P}/S$ to denote the corresponding partition in $G/S$, i.e., $\mathcal{P}/S = \{P/S: P\in \mathcal{P}\}$.

\smallskip\noindent{\bf Rank.}
We use $r(\cdot)$ to denote the rank function of the graphic matroid on graph $G=(V,E)$. For any edge set $F\subseteq E$, we have $r(F) = |V| - {\cal C}(F)$, where ${\cal C}(F)$ is the number of connected components in the graph $(V, F)$. 
The following properties are standard.
\begin{fact}\label{fact:rank-submodular}
    $r(\cdot)$ is submodular, i.e., for edge sets $A\subseteq B$ and $X$ disjoint from $B$,
    \[r(B\cup X) - r(B) \le r(A\cup X) - r(A)\]
\end{fact}
\begin{fact}\label{fact:rank-diff-cut-parts}
    Consider a graph $G=(V,E)$ and vertex subset $U\subseteq V$, and suppose a multiway cut $\mathcal{P}$ of $G$ partitions $U$ into $\{U_i\}_{i=1}^k$, and all induced subgraphs $G[U_i]$ as well as $G[U]$ are connected. Then,
    \[r(E[U]) = |U|-1, \quad
    r(E[U]\setminus \partial \mathcal{P}) = |U|-k.\]
    It follows that
    \[k-1 = r(E[U]) - r(E[U]\setminus \partial \mathcal{P}).\]
\end{fact}

\subsection{Cut Ratio and Skew-Density}
For any multiway cut $\mathcal{P}$, we call $\frac{d(\mathcal{P})}{|\mathcal{P}|-1}$ its {\em cut ratio}. Here, $|\cP|$ represents the number of subsets in the partition, i.e., the number of sides of $\cP$. A {\it minimum ratio cut} is a multiway cut with minimum cut ratio; its value is called the {\em strength} of the graph. An important property of a min-ratio cut is that all its sides are connected. 

\begin{fact}\label{fact:min-ratio-cut-connected}
Suppose $G$ is connected and $\mathcal{P}$ is the min-ratio cut of $G$. Then, the induced subgraph of each side $S\in \mathcal{P}$ is also connected.
\end{fact}
\begin{proof}
  Assume for contradiction that some $S\in \mathcal{P}$ is disconnected. Let $\{S_i\}_{i=1}^k, k\ge 2$ be the connected components of $S$. By replacing $S$ with $S_1, \ldots, S_k$ in $\mathcal{P}$, we get another multiway cut with the same cut edge set and strictly more sides. This implies the new cut will have strictly less ratio than $\mathcal{P}$, contradicting the assumption that $\mathcal{P}$ is a min ratio cut.
\end{proof}
For $S\subseteq V$ with $|S|\ge 2$, we define the \textit{skew-density} of $S$ as
\[\rho(S) = \frac{c(E[S])}{|S|-1}.\]
In particular, skew-density of singleton sets or the empty set is defined to be 0. 
Because singleton subgraphs have no edges, for any nonempty $S$ we have
\begin{equation}\label{eq:density-def-mult}
    c(E[S]) = \rho(S)\cdot (|S|-1).
\end{equation}

\subsection{Cut Hierarchy}

Consider the following recursive process. Given a connected graph $G$, find a min-ratio cut $\mathcal{P}$.  By \Cref{fact:min-ratio-cut-connected}, all sides of $\mathcal{P}$ are connected, so we can recursively run the same process on the induced subgraphs of all sides of $\mathcal{P}$. The recursion terminates at singleton subgraphs.
All sides of multiway cuts visited during this process form a laminar family.
We call this family the \emph{cut hierarchy} of $G$.

The cut hierarchy can be represented by a tree, where the root is $V$, and for each internal node $U$, all children of $U$ represent a min-ratio cut of $G[U]$. All leaves represent singleton sets.

\smallskip\noindent{\bf Canonical Cut Hierarchy.}
Note that the cut hierarchy is not necessarily unique because the min-ratio cut may not be unique. To avoid ambiguity in our algorithm, we always choose the min-ratio cut with most sides. We show that this min-ratio cut is unique, because of the following property:

\begin{lemma}\label{lem:min-ratio-cut-close-under-union}
    If $\mathcal{P}$ and $\mathcal{Q}$ are both min-ratio cuts, then the partition $\mathcal{W}$ formed by all connected components of $E\setminus(\partial \mathcal{P} \cup \partial \mathcal{Q})$ is also a min-ratio cut, and $\partial \mathcal{W} = \partial \mathcal{P} \cup \partial \mathcal{Q}$. 
\end{lemma}
\begin{proof}
    We first prove $\partial \mathcal{W} = \partial \mathcal{P} \cup \partial \mathcal{Q}$.
    \begin{enumerate}
        \item For any edge $(u, v)$ in $\partial \mathcal{W}$, $u$ and $v$ are disconnected in $E\setminus(\partial \mathcal{P} \cup \partial \mathcal{Q})$. The edge $(u, v)$ must be deleted to disconnect $u, v$, so $(u, v)\in \partial \mathcal{P} \cup \partial \mathcal{Q}$. We conclude that $\partial \mathcal{W} \subseteq \partial \mathcal{P} \cup \partial \mathcal{Q}$. 
        \item For any edge $(u, v)\in \partial \mathcal{P}$, $u$ and $v$ are in different sides of $\mathcal{P}$. Then, any $u$-$v$ path must contain an edge in $\partial \mathcal{P}$, and there cannot be a $u$-$v$ path after deleting $\partial \mathcal{P}$. In other words, $u$ and $v$ are disconnected in $E\setminus(\partial \mathcal{P} \cup \partial \mathcal{Q})$, and $(u, v)\in \partial \mathcal{W}$. We conclude that $\partial \mathcal{P} \subseteq \partial \mathcal{W}$. $\partial \mathcal{Q} \subseteq \partial \mathcal{W}$ follows the same argument, so $\partial \mathcal{P} \cup \partial \mathcal{Q} \subseteq \partial \mathcal{W}$. 
    \end{enumerate}

    Next, we prove that $\mathcal{W}$ is a min ratio cut.
    Let $\pi$ be the ratio of min ratio cuts $\mathcal{P}$ and $\mathcal{Q}$.
    It suffices to prove that the ratio of $\mathcal{W}$ is at most $\pi$.

    We start by applying the inclusion-exclusion principle to get
    \[c(\partial \mathcal{P}\cup \partial \mathcal{Q}) =  c(\partial \mathcal{P})+c(\partial \mathcal{Q}) - c(\partial \mathcal{P}\cap \partial \mathcal{Q})\]
    Notice that the first three terms are cut values. We have
    \[d(\mathcal{W}) = d(\mathcal{P})+d(\mathcal{Q}) - c(\partial \mathcal{P}\cap \partial \mathcal{Q})\]
    The last term may not be a cut value. We relate it to another partition as follows.
    Let $\mathcal{R}$ be the partition formed by all connected components in $E\setminus (\partial \mathcal{P}\cap \partial \mathcal{Q})$.
By the same argument as point 1 above, $\partial \mathcal{R}\subseteq \partial \mathcal{P}\cap \partial \mathcal{Q}$. We have
    \[c(\partial \mathcal{P}\cap \partial \mathcal{Q})\ge c(\partial \mathcal{R}) \ge \pi(|\mathcal{R}|-1)\]
    When $\mathcal{R}$ is a multiway cut, the second inequality above is by minimality of $\pi$; when $\mathcal{R}$ only has one component, the second inequality degenerates to $0\ge 0$.

    Recall that $\pi=\frac{d(\mathcal{P})}{|\mathcal{P}|-1} = \frac{d(\mathcal{Q})}{|\mathcal{Q}|-1}$.
    So,
     \[d(\mathcal{W}) = d(\mathcal{P})+d(\mathcal{Q}) - c(\partial \mathcal{P}\cap \partial \mathcal{Q}) \le \pi(|\mathcal{P}|-1) +\pi(|\mathcal{Q}|-1)-\pi(|\mathcal{R}|-1)\]
    We claim that $|\mathcal{P}|+|\mathcal{Q}|-|\mathcal{R}| \le |\mathcal{W}|$. Then, we have $d(\mathcal{W})\le \pi(|\mathcal{W}|-1)$, which means the ratio of $\mathcal{W}$ is at most $\pi$.

    It remains to establish the claim. Consider any side $R_i$ of $\mathcal{R}$. 
    Let $\mathcal{P}_i\subseteq\mathcal{P}$ (resp.\ $\mathcal{Q}_i\subseteq\mathcal{Q},\mathcal{W}_i\subseteq\mathcal{W}$) be the sides of $\mathcal{P}$ (resp.\ $\mathcal{Q}, \mathcal{W}$) that are contained in $R_i$. Because $\mathcal{P}$, $\mathcal{Q}$, $\mathcal{W}$ are subdivisions of $\mathcal{R}$ (since $\partial \mathcal{R}\subseteq \partial \mathcal{P},\partial \mathcal{R}\subseteq \partial \mathcal{Q},\partial \mathcal{R}\subseteq \partial \mathcal{W}$), the subsets $\mathcal{P}_i$ (resp.\ $\mathcal{Q}_i, \mathcal{W}_i$) form a partition of $\mathcal{P}$ (resp.\ $\mathcal{Q}, \mathcal{W}$). In particular, $\sum_i|\mathcal{P}_i|=|\mathcal{P}|$, $\sum_i|\mathcal{Q}_i|=|\mathcal{Q}|$, and $\sum_i|\mathcal{W}_i|=|\mathcal{W}|$.
    
    As a side of $\mathcal{R}$, $R_i$ is connected after deleting $\partial \mathcal{P}\cap \partial \mathcal{Q}$. By definition of $\mathcal{W}$, each side of $\mathcal{W}$ is connected. Let $T_i$ be a tree spanning $R_i$ in $E\setminus (\partial \mathcal{P}\cap \partial \mathcal{Q})$ satisfying $|T_i\cap(\partial\mathcal{W})|=|\mathcal{W}_i|-1$, which can be constructed by starting with a spanning tree in each $\mathcal {W}_i$ and then connecting the spanning trees to each other. 
    Because $T_i$ connects $|\mathcal{P}_i|$ components of $\mathcal{P}$, $|T_i\cap \partial \mathcal{P}| \ge |\mathcal{P}_i|-1$. By the same argument, $|T_i\cap \partial \mathcal{Q}| \ge |\mathcal{Q}_i|-1$.
    Because no edge of $T_i$ is in both $\partial \mathcal{P}$ and $\partial \mathcal{Q}$, the sets $T_i\cap \partial \mathcal{P}$ and $T_i\cap \partial \mathcal{Q}$ are disjoint. In conclusion
    \[|\mathcal{W}_i|-1=|T_i\cap\partial\mathcal{W}| = |T_i\cap(\partial\mathcal{P}\cup\partial\mathcal{Q})| = |T_i\cap \partial \mathcal{P}| + |T_i\cap \partial \mathcal{Q}| \ge |\mathcal{P}_i|-1 + |\mathcal{Q}_i|-1\]
    Suppose $\mathcal{R}=\{R_1, \ldots, R_k\}$. Summing over $i$,
    \[|\mathcal{W}| - k = \sum_{i=1}^k (|\mathcal{W}_i|-1) \ge \sum_{i=1}^k (|\mathcal{P}_i|-1+|\mathcal{Q}_i|-1) = |\mathcal{P}|-k+|\mathcal{Q}|-k\]
    which implies the claim.
    \eat{ Fix an arbitrary vertex $v\in V$, and let $\mathcal{P}_{1u}$ be the event that $\mathcal{P}$ separates vertices 1 and $u$. \textcolor{blue}{Jason: why is the first equality below true? $\sum_u1[\mathcal{P}_{1u}] $ is the number of vertices outside vertex $1$'s component, which is not necessarily $|\mathcal{P}|-1$ ? Ruoxu: The sum is over equivalence classes partitioned by P and Q}
     \[|\mathcal{P}|-1+|\mathcal{Q}|-1 
     \le \sum_{u\ne 1}1[\mathcal{P}_{1u}] + 1[\mathcal{Q}_{1u}]
     =\sum_{u\ne 1}1[\mathcal{P}_{1u}\lor \mathcal{Q}_{1u}] + 1[\mathcal{P}_{1u}\land \mathcal{Q}_{1u}]
     = \sum_{u\ne 1} 1 + 1[\mathcal{R}_{1u}]
    = n-1 + |\mathcal{R}|-1\]
    So, $|E|\le \pi(n-1)$. In conclusion, $E=\partial \mathcal{P}\cup\partial \mathcal{Q}$ is also a min ratio cut.}
\end{proof}

Let $\mathcal{W}$ be the connected components after deleting the edge sets of all min-ratio cuts. 
By \Cref{lem:min-ratio-cut-close-under-union}, $\mathcal{W}$ is also a min-ratio cut.
$\mathcal{W}$ maximizes the number of sides among all min-ratio cuts.
We call $\mathcal{W}$ the \emph{maximal min-ratio cut}.

We use \emph{canonical cut hierarchy} to denote the hierarchy constructed by the recursive definition given above, where in each step we choose the maximal min-ratio cut. From now on, we mean the canonical cut hierarchy when we refer to the cut hierarchy of a graph.

%% file: structure.tex
Of key importance to our cut hierarchy algorithm are the notions of star sets and dense cores. We define these in this section, and establish structural connections between them. We then algorithmically exploit these properties in our cut hierarchy algorithm, which appears in the next section.

\subsection{Star Sets and Dense Cores}

We define a \emph{star set} to be a vertex set that corresponds to an internal node of the canonical cut hierarchy whose children are all leaves. I.e., the maximal min-ratio cut of $G[S]$ is the all-singleton cut.

A single step of our cut hierarchy algorithm is to find a star set, which it can then contract and recurse. 
To find a star set, our algorithm finds a dense core instead, which is defined below.
Recall that $\rho(S)$ is the skew-density of $S$, defined as
$\rho(S) = \frac{c(E[S])}{|S|-1}$ when $|S|\ge 2$, and $\rho(S)=0$ when $|S|\le 1$.
A dense core is a vertex set that is skew-denser than all its subsets and supersets. Formally,

\begin{defn}
    A vertex set $S\subseteq V$ is called a \emph{dense-core} if $\forall W\subseteq S, \rho(W)\le \rho(S)$, and $\forall U\supsetneqq S, \rho(U)<\rho(S)$. 
\end{defn}

We say a vertex set is a skew-densest set if it has maximum skew-density among all vertex sets of the graph. We say a vertex set is a maximum skew-densest set if it maximizes size among all skew-densest sets. The following fact is straightforward to check from the definition.
\begin{fact}\label{fact:core-exist}
    A maximum skew-densest set is a dense core. Hence, every graph contains a dense core.
\end{fact}

The following fact shows that the induced subgraph of a dense core is connected. 
\begin{fact}\label{fact: core-connected}
    If $S$ is a dense core in a nonempty graph $G$, then $G[S]$ is connected.
\end{fact}
\begin{proof}
 Let $\{S_i\}_{i=1}^k, k\ge 1$ be the connected components of $S$. Because $S$ is a dense-core and $S_i\subseteq S$, we have $\rho(S_i)\le \rho(S)$ for each $S_i$. Combined with (\ref{eq:density-def-mult}), we have for all $i$,
    \[c(E[S_i]) = \rho(S_i)\cdot (|S_i|-1) \le \rho(S)\cdot (|S_i|-1)\]
    Taking the sum gives
    \[c(E[S]) = \sum_{i=1}^k c(E[S_i]) \le \rho(S)\cdot \bigg(\sum_{i=1}^k (|S_i| - 1)\bigg) = \rho(S)\cdot (|S|-k)\]
    Combined with $c(E[S])=\rho(S)\cdot (|S|-1)$ from (\ref{eq:density-def-mult}), we have 
    \[\rho(S)(|S|-1)\le \rho(S)(|S|-k) \implies \rho(S)(k-1)\le 0\]
    We assume the graph $G$ is nonempty, so $\rho(V)>0$. Either $V=S$ or $V\supsetneqq S$, in both cases we have $\rho(S)\ge \rho(V)$. Then $\rho(S)> 0$, and we conclude $k-1\le 0$. Because $k$ is a positive integer, we have $k=1$, which means $S$ is connected.
\end{proof}

\subsection{Dense cores are Star Sets}

Next, we show that dense cores and star sets are closely related.
We first show the relationship between a dense core and the min-ratio cut.

\begin{figure}
    \centering
    \input{fig1}
    \caption{Notations when $U$ crosses $\mathcal{P}$}
    \label{fig:core-is-star-onestep}
\end{figure}
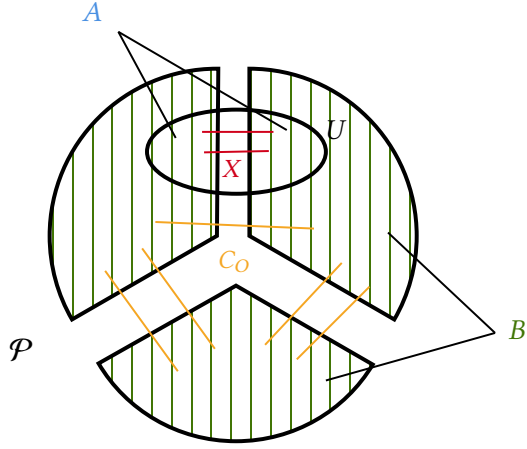
\begin{lemma}\label{lem:core-is-star-onestep}
    Let $\mathcal{P}$ be a min-ratio cut in a connected graph $G$. Suppose $U$ is a dense core in $G$ and $U\ne V$. Then, $U$ is contained in one side of $\mathcal{P}$.
\end{lemma}
\begin{proof}
    Assume for contradiction that $G[U]$ is separated into more than one side of $\mathcal{P}$. Because $G[U]$ is connected by \Cref{fact: core-connected}, this means $E[U]\cap \partial \mathcal{P}\ne\emptyset$.

    Let $r(\cdot)$ be the rank function of the graphic matroid.
    Recall that $r$ is submodular. So, for edge sets $A\subseteq B$ and $X$ disjoint from $B$, we have
    \[r(B\cup X) - r(B) \le r(A\cup X) - r(A)\]

    Substituting $A = E[U]\setminus \partial \mathcal{P}, B = E\setminus \partial \mathcal{P}, X = E[U]\cap \partial \mathcal{P}$ and denoting $C_O = \partial \mathcal{P}\setminus X$ (see \Cref{fig:core-is-star-onestep}), we have 
    \begin{equation}\label{eq:rank-submodular}
        r(E\setminus C_O) - r(E\setminus \partial \mathcal{P})  \le r(E[U]) - r(E[U]\setminus \partial \mathcal{P})
    \end{equation}

   In the min-ratio cut $\mathcal{P}$, every side is connected by \Cref{fact:min-ratio-cut-connected}. Applying \Cref{fact:rank-diff-cut-parts} on $U=V$ gives $|\mathcal{P}|-1 = r(E) - r(E\setminus \partial \mathcal{P})$.
  Denote $\pi$ to be the ratio of $\mathcal{P}$, i.e.,
    \[d(\mathcal{P}) = \pi\cdot (r(E)-r(E\setminus \partial \mathcal{P}))\]
    We claim
    \[c(C_O) \ge \pi \cdot (r(E) - r(E\setminus C_O))\]
    When $C_O$ does not separate anything in $G$, RHS is trivially 0. Otherwise, a subset of $C_O$ is the edge set of a multiway cut on $G$, whose ratio is at least $\pi$, establishing the claim.
    
    Taking the difference between the above two inequalities and applying (\ref{eq:rank-submodular}) gives
    \begin{equation}\label{eq:X-ratio}
        c(X) = d(\mathcal{P}) - c(C_O)
        \le \pi \cdot (r(E\setminus C_O)-r(E\setminus  \partial \mathcal{P})) \le \pi \cdot (r(E[U]) - r(E[U]\setminus  \partial \mathcal{P}))
    \end{equation}

    We assumed for contradiction that $U$ is separated into more than one side of $\mathcal{P}$.
    Let $\{U_i\}_{i=1}^k, k\ge 2$ be the connected components in $G[U]$ after deleting all cut edges $\partial \mathcal{P}$.
    Because $U$ is a dense core and $U_i$ are subsets of $U$, $\rho(U_i) \le \rho(U)$ for each $U_i$. Combined with (\ref{eq:density-def-mult}), we have for all $i$,
    \[c(E[U_i]) = \rho(U_i)\cdot(|U_i|-1) \le \rho(U) \cdot (|U_i|-1)\]
    Summing over all $i$,
    \begin{equation}\label{eq:U-minus-C}
        c(E[U]\setminus \partial \mathcal{P}) = \sum_{i=1}^k c(E[U_i]) \le \rho(U) \cdot (|U|-k) = \rho(U) \cdot r(E[U]\setminus \partial \mathcal{P})
    \end{equation}
    where the last equality is by \Cref{fact:rank-diff-cut-parts}.
    Combining (\ref{eq:X-ratio}) and (\ref{eq:U-minus-C}), we have
    \[c(E[U])=c(E[U]\setminus \partial \mathcal{P}) + c(X)
    \le \rho(U) \cdot r(E[U]\setminus \partial \mathcal{P}) + \pi \cdot (r(E[U]) - r(E[U]\setminus \partial \mathcal{P})) \]
    Notice that $\pi \le \rho(V)$ because $\rho(V)$ is the ratio of the all-singleton cut. Because $U$ is a dense core and $V\supsetneqq U$, we have $\rho(U)> \rho(V) \ge \pi$. By the assumption that $\mathcal{P}$ partitions $G[U]$ into multiple connected components, $r(E[U]) > r(E[U]\setminus \partial \mathcal{P})$. Therefore,
    \begin{align*}
        \rho(U)\cdot(|U|-1)=c(E[U]) &<  \rho(U) \cdot r(E[U]\setminus \partial \mathcal{P}) +\rho(U)  \cdot (r(E[U]) - r(E[U]\setminus \partial \mathcal{P})) \\
    &= \rho(U) \cdot r(E[U]) 
    \end{align*}
    Finally, noticing that $r(E[U]) = |U|-1$ because $G[U]$ is connected, we have $p(U)\cdot(|U|-1) < p(U)\cdot(|U|-1)$, which is a contradiction.
\end{proof}

\begin{lemma}\label{lem:core-is-star}
    If $S$ is a dense core, then $S$ is a star set of the canonical cut hierarchy.
\end{lemma}
\begin{proof}
    We prove by induction on $n=|V|$.
    
    The base case is $n=|S|$ or equivalently $V=S$. It suffices to show that the all-singleton cut is a min-ratio cut of $G$. If this is true, then the recursive min-ratio cut procedure will choose the all-singleton cut which is the maximal min-ratio cut of $G$. Then $S$ is partitioned into singletons, which means $S$ is a star set in the hierarchy.
    
    Because $S$ is a dense core, for all nonempty $W\subseteq S$,
    \[c(E[W]) = \rho(W) \cdot (|W|-1) \le \rho(S)\cdot (|W|-1) \]
    Consider any multiway cut $\mathcal{P}=\{W_i\}_{i=1}^k$ of $G[S]$. We have
    \[d(\mathcal{P}) = c(E[S]) - \sum_{i=1}^k c(E[W_i])
    \ge \rho(S)\cdot (|S|-1) - \rho(S)\cdot \bigg(\sum_{i=1}^k (|W_i|-1)\bigg) = \rho(S)\cdot(k-1) \]
    So the ratio of $\mathcal{P}$ is at least $\rho(S)$, which is the ratio of the all-singleton cut. In conclusion the all-singleton cut is a min-ratio cut.

    Next we consider the inductive case where $S\subsetneqq V$. Let $\mathcal{P}$ be the maximal min-ratio cut of $G$.
    By \Cref{lem:core-is-star-onestep}, $S$ is contained in one side $U\in \mathcal{P}$.
    Note that since $S$ is a dense core in $G$, $S$ is also a dense core in $G[U]$ for $U\subseteq V$ because the conditions on graph $G[U]$ is a subset of the conditions on $G$. Using the inductive hypothesis, we get that $S$ is a star set in the canonical hierarchy of $G[U]$, which is a subtree of the canonical hierarchy of $G[V]$. In conclusion $S$ is a star set in the canonical hierarchy of $G$.
\end{proof}

%% file: fig1.tex
 
\tikzset{
pattern size/.store in=\mcSize, 
pattern size = 5pt,
pattern thickness/.store in=\mcThickness, 
pattern thickness = 0.3pt,
pattern radius/.store in=\mcRadius, 
pattern radius = 1pt}
\makeatletter
\pgfutil@ifundefined{pgf@pattern@name@_7zh3hxdjf}{
\pgfdeclarepatternformonly[\mcThickness,\mcSize]{_7zh3hxdjf}
{\pgfqpoint{-\mcThickness}{-\mcThickness}}
{\pgfpoint{\mcSize}{\mcSize}}
{\pgfpoint{\mcSize}{\mcSize}}
{
\pgfsetcolor{\tikz@pattern@color}
\pgfsetlinewidth{\mcThickness}
\pgfpathmoveto{\pgfpointorigin}
\pgfpathlineto{\pgfpoint{0}{\mcSize}}
\pgfusepath{stroke}
}}
\makeatother

 
\tikzset{
pattern size/.store in=\mcSize, 
pattern size = 5pt,
pattern thickness/.store in=\mcThickness, 
pattern thickness = 0.3pt,
pattern radius/.store in=\mcRadius, 
pattern radius = 1pt}
\makeatletter
\pgfutil@ifundefined{pgf@pattern@name@_6cgr2dufv}{
\pgfdeclarepatternformonly[\mcThickness,\mcSize]{_6cgr2dufv}
{\pgfqpoint{-\mcThickness}{-\mcThickness}}
{\pgfpoint{\mcSize}{\mcSize}}
{\pgfpoint{\mcSize}{\mcSize}}
{
\pgfsetcolor{\tikz@pattern@color}
\pgfsetlinewidth{\mcThickness}
\pgfpathmoveto{\pgfpointorigin}
\pgfpathlineto{\pgfpoint{0}{\mcSize}}
\pgfusepath{stroke}
}}
\makeatother

 
\tikzset{
pattern size/.store in=\mcSize, 
pattern size = 5pt,
pattern thickness/.store in=\mcThickness, 
pattern thickness = 0.3pt,
pattern radius/.store in=\mcRadius, 
pattern radius = 1pt}
\makeatletter
\pgfutil@ifundefined{pgf@pattern@name@_w1m9gvu4f}{
\pgfdeclarepatternformonly[\mcThickness,\mcSize]{_w1m9gvu4f}
{\pgfqpoint{-\mcThickness}{-\mcThickness}}
{\pgfpoint{\mcSize}{\mcSize}}
{\pgfpoint{\mcSize}{\mcSize}}
{
\pgfsetcolor{\tikz@pattern@color}
\pgfsetlinewidth{\mcThickness}
\pgfpathmoveto{\pgfpointorigin}
\pgfpathlineto{\pgfpoint{0}{\mcSize}}
\pgfusepath{stroke}
}}
\makeatother

 
\tikzset{
pattern size/.store in=\mcSize, 
pattern size = 5pt,
pattern thickness/.store in=\mcThickness, 
pattern thickness = 0.3pt,
pattern radius/.store in=\mcRadius, 
pattern radius = 1pt}
\makeatletter
\pgfutil@ifundefined{pgf@pattern@name@_j64q8mp8h}{
\pgfdeclarepatternformonly[\mcThickness,\mcSize]{_j64q8mp8h}
{\pgfqpoint{0pt}{0pt}}
{\pgfpoint{\mcSize+\mcThickness}{\mcSize+\mcThickness}}
{\pgfpoint{\mcSize}{\mcSize}}
{
\pgfsetcolor{\tikz@pattern@color}
\pgfsetlinewidth{\mcThickness}
\pgfpathmoveto{\pgfqpoint{0pt}{0pt}}
\pgfpathlineto{\pgfpoint{\mcSize+\mcThickness}{\mcSize+\mcThickness}}
\pgfusepath{stroke}
}}
\makeatother

 
\tikzset{
pattern size/.store in=\mcSize, 
pattern size = 5pt,
pattern thickness/.store in=\mcThickness, 
pattern thickness = 0.3pt,
pattern radius/.store in=\mcRadius, 
pattern radius = 1pt}
\makeatletter
\pgfutil@ifundefined{pgf@pattern@name@_f2ffvty71}{
\pgfdeclarepatternformonly[\mcThickness,\mcSize]{_f2ffvty71}
{\pgfqpoint{0pt}{0pt}}
{\pgfpoint{\mcSize+\mcThickness}{\mcSize+\mcThickness}}
{\pgfpoint{\mcSize}{\mcSize}}
{
\pgfsetcolor{\tikz@pattern@color}
\pgfsetlinewidth{\mcThickness}
\pgfpathmoveto{\pgfqpoint{0pt}{0pt}}
\pgfpathlineto{\pgfpoint{\mcSize+\mcThickness}{\mcSize+\mcThickness}}
\pgfusepath{stroke}
}}
\makeatother
\tikzset{every picture/.style={line width=0.75pt}} 

\begin{tikzpicture}[x=0.75pt,y=0.75pt,yscale=-1,xscale=1]

\draw  [pattern=_7zh3hxdjf,pattern size=7.5pt,pattern thickness=0.75pt,pattern radius=0pt, pattern color={rgb, 255:red, 65; green, 117; blue, 5}][line width=1.5]  (393.63,200.62) .. controls (379.72,223.79) and (354.04,239.33) .. (324.67,239.33) .. controls (295.05,239.33) and (269.2,223.54) .. (255.36,200.06) -- (324.67,160.8) -- cycle ;
\draw  [pattern=_6cgr2dufv,pattern size=7.5pt,pattern thickness=0.75pt,pattern radius=0pt, pattern color={rgb, 255:red, 65; green, 117; blue, 5}][line width=1.5]  (241.96,177.88) .. controls (234.78,165.46) and (230.67,151.05) .. (230.67,135.68) .. controls (230.67,89.08) and (268.44,51.3) .. (315.05,51.3) .. controls (315.28,51.3) and (315.51,51.3) .. (315.74,51.3) -- (315.05,135.68) -- cycle ;
\draw  [pattern=_w1m9gvu4f,pattern size=7.5pt,pattern thickness=0.75pt,pattern radius=0pt, pattern color={rgb, 255:red, 65; green, 117; blue, 5}][line width=1.5]  (331.42,51.3) .. controls (331.42,51.3) and (331.42,51.3) .. (331.42,51.3) .. controls (378.02,51.3) and (415.8,89.08) .. (415.8,135.68) .. controls (415.8,151.05) and (411.68,165.46) .. (404.5,177.88) -- (331.42,135.68) -- cycle ;
\draw  [line width=1.5]  (279.78,93.2) .. controls (279.78,81.47) and (299.98,71.96) .. (324.89,71.96) .. controls (349.8,71.96) and (370,81.47) .. (370,93.2) .. controls (370,104.92) and (349.8,114.43) .. (324.89,114.43) .. controls (299.98,114.43) and (279.78,104.92) .. (279.78,93.2) -- cycle ;
\draw [color={rgb, 255:red, 208; green, 2; blue, 27 }  ,draw opacity=1 ]   (307.56,83.2) -- (343.56,83.2) ;
\draw [color={rgb, 255:red, 208; green, 2; blue, 27 }  ,draw opacity=1 ]   (308.56,93.53) -- (341.22,92.86) ;
\draw  [draw opacity=0][pattern=_j64q8mp8h,pattern size=6pt,pattern thickness=0.75pt,pattern radius=0pt, pattern color={rgb, 255:red, 74; green, 144; blue, 226}] (315.91,111.45) .. controls (313.78,111.81) and (311.54,112) .. (309.25,112) .. controls (293.95,112) and (281.56,103.55) .. (281.56,93.13) .. controls (281.56,82.71) and (293.95,74.27) .. (309.25,74.27) .. controls (311.54,74.27) and (313.78,74.46) .. (315.91,74.82) -- cycle ;
\draw  [draw opacity=0][pattern=_f2ffvty71,pattern size=6pt,pattern thickness=0.75pt,pattern radius=0pt, pattern color={rgb, 255:red, 74; green, 144; blue, 226}][line width=1.5]  (331.97,75.08) .. controls (334.86,74.38) and (337.98,74) .. (341.22,74) .. controls (356.52,74) and (368.91,82.44) .. (368.91,92.86) .. controls (368.91,103.28) and (356.52,111.73) .. (341.22,111.73) .. controls (337.52,111.73) and (333.99,111.23) .. (330.76,110.33) -- cycle ;
\draw    (265.33,32.67) -- (294.27,86.27) ;
\draw    (265.33,32.67) -- (350.27,82.27) ;
\draw [color={rgb, 255:red, 245; green, 166; blue, 35 }  ,draw opacity=1 ]   (277.33,142) -- (314,192.67) ;
\draw [color={rgb, 255:red, 245; green, 166; blue, 35 }  ,draw opacity=1 ]   (378,149.33) -- (339.33,190.67) ;
\draw [color={rgb, 255:red, 245; green, 166; blue, 35 }  ,draw opacity=1 ]   (284,128.67) -- (364,132) ;
\draw [color={rgb, 255:red, 245; green, 166; blue, 35 }  ,draw opacity=1 ]   (392,161.33) -- (355.33,198) ;
\draw [color={rgb, 255:red, 245; green, 166; blue, 35 }  ,draw opacity=1 ]   (258.67,153.33) -- (295.33,204) ;
\draw    (401.33,134) -- (455.33,184.67) ;
\draw    (370,208.67) -- (455.33,184.67) ;

\draw (368.67,75.93) node [anchor=north west][inner sep=0.75pt]  [font=\normalsize]  {$U$};
\draw (208.67,186.07) node [anchor=north west][inner sep=0.75pt]  [font=\large]  {$\mathcal{P}$};
\draw (316.67,95.73) node [anchor=north west][inner sep=0.75pt]  [font=\small,color={rgb, 255:red, 208; green, 2; blue, 27 }  ,opacity=1 ]  {$X$};
\draw (246,16.07) node [anchor=north west][inner sep=0.75pt]  [color={rgb, 255:red, 74; green, 144; blue, 226 }  ,opacity=1 ]  {$A$};
\draw (314.67,141.73) node [anchor=north west][inner sep=0.75pt]  [font=\small,color={rgb, 255:red, 245; green, 166; blue, 35 }  ,opacity=1 ]  {$C_{O}$};
\draw (460.67,176.4) node [anchor=north west][inner sep=0.75pt]    {$\textcolor[rgb]{0.25,0.46,0.02}{B}$};

\end{tikzpicture}

%% file: algorithm.tex
In this section, we design an algorithm (\Cref{alg:hierarchy}) that constructs the canonical cut hierarchy w.h.p.\ in $\tO(nm\log C)$ time plus $\tO(n\log C)$ max-flow calls, thereby establishing \Cref{thm:cut-hierarchy}.

\cuthierarchyAlg*

In each iteration of the algorithm, we identify and contract a star set in the canonical cut hierarchy of the current graph (see \Cref{fig:alg-illustration}). Recall that a star set represents a node in the hierarchy whose children are all leaves of the hierarchy. In other words, it is a minimal non-singleton set in the laminar family corresponding to the hierarchy.
We will show that contracting a star set does not affect the rest of the laminar family. Then, we can obtain the laminar family by collecting all sets found by the algorithm.

\begin{figure}
    \centering
    \includegraphics{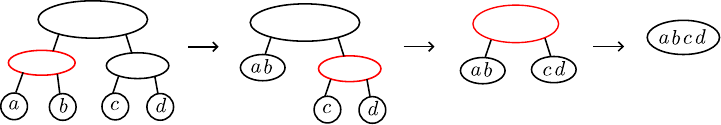}
    \caption{Illustration of iterative contraction algorithm. Red nodes represent the contracted star set in the hierarchy.}
    \label{fig:alg-illustration}
\end{figure}

\subsection{Contracting Star Sets}
We first observe that contracting a subset of a side does not affect the maximal min-ratio cut.
\begin{fact}\label{fact:contract-aligned-set-preserve-min-ratio-cut}
    Suppose $\mathcal{P}$ is the maximal min-ratio cut of $G$ with a side $U\in\mathcal{P}$, and $S\subseteq U$.
    Then, $\mathcal{P}/S$ is still the maximal min-ratio cut of $G/S$ after contracting $S$.
\end{fact}
\begin{proof}
Let $\mathcal{C}_1$ be the family of multiway cuts in $G$ that do not separate $S$, and $\mathcal{C}_2$ be the family of multiway cuts in $G/S$.
The map $\phi:\mathcal{C}_1 \to \mathcal{C}_2$ defined by $\phi(\mathcal{Q})=\mathcal{Q}/S$ is a bijection. Note that $\phi$ preserves cut ratio because $\phi$ preserves the set of cut edges and the number of sides.

Because $\mathcal{P}\in \mathcal{C}_1$ and $\mathcal{P}$ is the maximal min-ratio cut of $G$, $\mathcal{P}$ maximizes the number of sides among the multiway cuts that minimize cut ratio in $\mathcal{C}_1$. After applying $\phi$, this property is preserved by $\mathcal{P}/S$ in $\mathcal{C}_2$, which means $\mathcal{P}/S$ is the maximal min-ratio cut of $G/S$.
\end{proof}

A simple corollary of this fact is that contracting a star set does not affect the rest of the laminar family. This validates our bottom-up contraction algorithm.

\begin{corollary}\label{lem:star-contraction-preserve-hierarchy}
    If we contract a star set $S$, then in the canonical cut hierarchy, the only change is that the node corresponding to $S$ is replaced by a leaf corresponding to the contracted vertex.
\end{corollary}
\begin{proof}
    We compare the recursive min-ratio cut process before and after the contraction. On a set $U$ disjoint from $S$, it will also choose the same maximal min-ratio cut. On a set $U\supsetneqq S$, it will choose the same maximal min-ratio cut up to contracting $S$ by \Cref{fact:contract-aligned-set-preserve-min-ratio-cut}.  The only difference is that originally the process recurses on $S$, while after contraction, $S$ is a singleton and hence a leaf in the hierarchy.
\end{proof}


\subsection{Algorithm Framework}
\label{sec:framework}
The algorithm runs in $O(n)$ iterations. Each iteration finds a star set and contracts it. The output is the laminar family of all sets found during the algorithm.

We give a high-level description before formally stating the algorithm. Recall that by \Cref{fact:core-exist}, the maximum skew-densest set is  a dense core, and hence a star set by \Cref{lem:core-is-star}. We can design an algorithm for finding the maximum skew-densest subgraph (call it $D$) based on previously known approaches to finding the densest subgraph. The bottleneck is that this algorithm uses a directed (global) min-cut subroutine, for which the best known algorithm has a running time of $m^{1+o(1)} \sqrt{n}$~\cite{CLNPQS21}. This is sufficient if $|D| \ge \Omega(\sqrt{n})$ since we can amortize the running time over the sets3 that we contract. But, for small $D$, we need a sharper running time bound. For this purpose, we show that a modification of the global directed min-cut algorithm can be used to identify $D$ in $m^{1+o(1)} k$ time, for a given upper bound $k \ge |D|$. So, if we are given an upper bound $k$ that (say) satisfies $k \le 2|D|$, we would be done. But, how we find such an upper bound? We can use guess-and-double, but for this to work, we need to identify the complementary case: $k < |D|$. In other words, the above algorithm returns some set as the presumptive $D$, and we need to verify if that is indeed the case. Unfortunately, we do not know how to do this efficiently. Instead, we settle for the simpler goal: identify if the set returned by the above algorithm is a dense core. For this latter task, we give an algorithm that runs in $m^{1+o(1)} k$ time. Note that this is sufficient since any dense core (whether or not the maximum skew-densest set) is a star set in the cut hierarchy.

Now, we formally describe the algorithm for each iteration (see \Cref{alg:hierarchy}).
The algorithm guesses a size parameter $k$ by doubling (starting from 2, i.e.\ $k=2, 4, 8, \ldots$, because star sets are not singleton).
We call a subroutine {\sc Find-Star} described in \Cref{sec:find-star} to find a vertex set $S_k$. If $|S_k|\le k/2$ or $|S_k| > k$, we reject it and go to the next value of $k$. Otherwise, we run a subroutine {\sc Verify-Core} described in \Cref{sec:verify-star} to decide whether $S_k$ is a dense core. If $S_k$ is verified as a dense core, we finish the iteration and return $S_k$. Otherwise, we go to the next value of $k$.

\begin{algorithm}
\caption{{\sc Construct-Canonical-Min-Ratio-Cut-Hierarchy}($G$)}
\label{alg:hierarchy}
$\mathcal{F}\gets \emptyset$\\
\While{$G$ is not contracted into a singleton}{
    \For{$k=2,4,8,\ldots, 2^{\lceil \log n\rceil}$}{
    $S_k\gets \text{\sc Find-Star}(G, k)$\\
    \If{$k/2< |S_k| \le k$ and $\text{\sc Verify-Core}(G, k, S_k)$}{
    Add $S_k$ to $\mathcal{F}$\\
    Contract $S_k$ in $G$\\
    Break for loop
    }
    }
}
Output the laminar family $\mathcal{F}$.
\end{algorithm}

In the rest of the section, we describe the two subroutines of finding a candidate set (\Cref{alg:find} in \Cref{sec:find-star}) and verifying a dense core  (\Cref{alg:verify} in \Cref{sec:verify-star}).
The results are summarized in \Cref{lem:alg-star} and \Cref{lem:alg-verify}.
Informally, we guarantee that {\sc Find-Star} can find the maximum skew-densest set when $k$ is large enough, and {\sc Verify-Core} can correctly verify whether the input set is a dense core.

\begin{restatable}[Find-Star]{lemma}{FindStar}
\label{lem:alg-star}
    Given a connected input graph $G$ with integer edge capacities at most $C$ and a parameter $k$,  \Cref{alg:find} outputs the maximum skew-densest set $D$ assuming $k\ge|D|$. The algorithm runs in  $\tO(mk\log C)$ time plus $\tO(\log C)$ calls to max-flow on digraphs with $O(m)$ nodes and $O(m)$ edges.
\end{restatable}

\begin{restatable}[Verify-Core]{lemma}{VerifyCore}
\label{lem:alg-verify}
For an input graph $G$, a set of vertices $S$, and a parameter $k$, \Cref{alg:verify} returns true if and only if $S$ is a dense core in $G$ and $|S| \le k$. The running time is $O(k)$ calls to max-flow on digraphs with $O(m)$ nodes and $O(m)$ edges.
\end{restatable}

In the running time analysis, we use $F(n, m)$ to denote the time complexity of max-flow on digraph of $n$ nodes and $m$ edges. We assume $F(n, m)$ is monotone increasing and $F(n, m)\ge \Omega(m)$. The following is our main lemma, which establishes \Cref{thm:cut-hierarchy}, and follows from \Cref{lem:alg-star} and \Cref{lem:alg-verify}.

\begin{lemma}
    \Cref{alg:hierarchy} outputs the canonical cut hierarchy of $G$ w.h.p.\ in $\tO(nm\log C)$ time plus $\tO(n\log C)$ calls to max flow on digraphs with $O(m)$ edges and $O(m)$ vertices and integer edge weights at most $\poly(n, C)$.
\end{lemma}
\begin{proof}
    We first show that each iteration can find and contract a set $S_k$ w.h.p. Let $D$ be the maximum skew-densest set of the current contracted graph, which is a dense core by \Cref{fact:core-exist}. Let $k_D$ be the smallest power of $2$ greater or equal to $|D|$. If the inner for loop breaks before $k_D$, the claim holds. Otherwise, consider the inner for loop when $k=k_D$.
    By \Cref{lem:alg-star}, {\sc Find-Star} guarantees to find $D$ w.h.p.\ when $k\ge |D|$. By \Cref{lem:alg-verify}, $D$ will be verified as a dense core. Also, $k_D/2 < |D| \le k_D$. So, the inner for loop will break when $k=k_D$. In conclusion, the inner for loop will break before or at $k_D$, which means some $S_k$ is contracted.

    Suppose the above high probability event happens. Then, each iteration can only terminate when {\sc Verify-Core} returns true for some set $S_k$. By \Cref{lem:alg-verify}, such an $S_k$ must be a dense core. Then, $S_k$ is a star set by \Cref{lem:core-is-star}.
    Thus, each iteration contracts a star set, and this does not affect the rest of the hierarchy by \Cref{lem:star-contraction-preserve-hierarchy}. After $O(n)$ iterations, the whole graph is contracted into a singleton, and all sets visited by the algorithm form the canonical cut hierarchy.

    Each iteration runs $O(\log n)$ steps of doubling $k$. For each $k$, the algorithm calls {\sc Find-Star} and {\sc Verify-Core}, which take $\tO(mk\log C + F(m, m)(k+\log C)) = \tO((m+F(m, m))k\log C)$ time by \Cref{lem:alg-star,lem:alg-verify}. The iteration contracts a set of size $> k/2$, hence decreases the vertex size of the contracted graph by at least $k/2$. So, the amortized time to contract a vertex is $\tO((m+F(m, m))\log C)$, and the total running time is $\tO(nm
    \log C + n\log C\cdot F(m, m))$. (The infinite edge weights can be simulated by a large weight of $O(mC)$.)
\end{proof}

\subsection{Goldberg's Network and Gabow's Variant}
\label{sec:networks}
 We need Goldberg's network and its variant due to Gabow as key building blocks of both subroutines {\sc Find-Star} and {\sc Verify-Core}.

Given an undirected graph $G$ with edge capacity $c$ and a parameter $\tau > 0$, we construct a directed graph $H$ that we call Goldberg's network as follows \cite{Goldberg84} (See \Cref{fig:goldberg}). To avoid ambiguity, we use \emph{nodes} and \emph{arcs} in place of vertices and edges for Goldberg's network $H$.
\begin{itemize}
 \item Start with a bipartite graph on node set $E\uplus V$ defined as follows: For each $e\in E$ and each endpoint $v$ of $e$ in $G$, add an arc in $H$ from node $e$ to node $v$ of capacity $\infty$.
 \item Add a node $s$ with an arc $(s,e)$ of capacity $c_e$ for each $e\in E$.
 \item Add a node $t$ with an arc $(v,t)$ of capacity $\tau$ for each $v\in V$.
\end{itemize}


\begin{figure}
    \centering
    \includegraphics{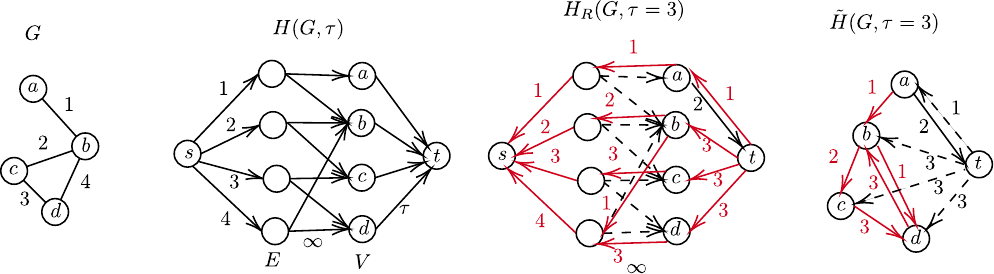}
    \caption{Example of $H$ and $\tilde{H}$ on input graph $G$ and $\tau=3$}
    \label{fig:goldberg}
\end{figure}

\begin{fact}\label{fact:H-size}
    $|V[H]| = O(m+n),\, |E[H]| = O(m+n)$.
\end{fact}

In this section, we require the concepts of (directed) $s$-$t$ flows, $s$-$t$ cuts, and $t$-cuts of directed graphs. An $s$-$t$ flow assigns a value $f(u,v)$ to each arc $(u,v)$ satisfying $0\le f(u,v)\le c(u,v)$. For simplicity, assume that $f(u,v)=0$ if $(u,v)$ is not an arc of the graph. A $s$-$t$ flow must satisfy the flow conservation constraints, namely $\sum_vf(u,v)=\sum_vf(v,u)$ for all vertices $u\notin\{s,t\}$. The value of the $s$-$t$ flow equals $\sum_vf(s,v)$, and the $s$-$t$ max flow is the $s$-$t$ flow of maximum value. Given an $s$-$t$ flow $f$, the residual graph of $f$ has the following capacities on arcs: for each pair of arcs $(u,v)$ and $(v,u)$, where a nonexistent arc is considered an arc of capacity $0$, the residual graph sets capacities $c(u,v)-f(u,v)+f(v,u)$ and $c(v,u)-f(v,u)+f(u,v)$ on the arcs $(u,v)$ and $(v,u)$, respectively.

We use $\partial^+S$ to denote the set of arcs from $S$ to $V\setminus S$.
An $s$-$t$ cut is the arc set $\partial^+S$ for some vertex set $S$ that contains $s$ but not $t$. We call $S$ the source side of the $s$-$t$ cut or $t$-cut. The value of $s$-$t$ cut $\partial^+S$ is the total capacity $d^+(S)=\sum_{e\in \partial^+S} c_e$. The $s$-$t$ mincut is the $s$-$t$ cut of minimum value, which is also the value of the $s$-$t$ max flow. A $t$-cut is the arc set $\partial^+S$ for some set of vertices $S$ not containing $t$. We also call $S$ the source side of the $t$-cut, and value and $t$-mincut are defined analogously. We sometimes abuse notation and refer to the $s$-$t$ cut or $t$-cut by the source side $S$ instead of the set of arcs $\partial^+S$.

We require the following fact about cuts in residual graphs, which follows from flow conservation and the construction of the residual graph.
\begin{fact}\label{fact:residual}
Consider a directed graph $H$, an $s$-$t$ flow of value $\tau$, and the residual graph $H_R$ of $f$. For any vertex set $S$ that contains $s$ but not $t$, we have $d^+_{H_R}(S)=d^+_H(S)-\tau$.
\end{fact}

We can characterize the value of any (directed) $s$-$t$ cut on Goldberg's network $H$ as follows.
\begin{fact}\label{fact:goldberg-cut-value}
    For any $s$-$t$ cut on $H$ with source side $S\ni s$, let $S_E=S\cap E$ and $S_V=S\cap V$ be the nodes in the source side from $E$ and from $V$ respectively, i.e., $S=\{s\}\cup S_E\cup S_V$. Then, the cut value of $S$ is
    $$d^+_H(S) = \begin{cases} c(E) - c(S_E) + \tau\cdot |S_V| & S_E \subseteq E[S_V]\\
    \infty & S_E \not\subseteq E[S_V]
    \end{cases}$$
\end{fact}
\begin{proof}
    We first prove that the cut value is finite if and only if $S_E\subseteq E[S_V]$, i.e., $S_E$ is a subset of edges in the induced subgraph of $G$ on $S_V$.

    Assume the cut value is finite. For every $e=(u, v)\in S_E$, it has two outgoing infinite arcs $(e, u)$ and $(e, v)$. Infinite arcs cannot appear in $\partial^+S$, so $u, v\in S_V$. It follows that $(u, v)\in E[S_V]$. Since this holds for every $e\in S_E$, we have $S_E\subseteq E[S_V]$.

    Assume the cut value is infinite, i.e., there exists an infinite arc $(e, u)$ in $\partial^+S$. Then, $e=(u, v)$ for some vertex $v$, and $e\in S_E, u\notin S_V$. This implies $e\notin E[S_V]$ and $S_E\not\subseteq E[S_V]$.

    Next, we compute the cut value when it is finite. There are two types of cut arcs: from $s$ to $E\setminus S_E$ and from $S_V$ to $t$. The first type contributes $c(E\setminus S_E) = c(E)-c(S_E)$. The second type contributes $\tau\cdot |S_V|$. Taking their sum gives the statement.
\end{proof}

As a corollary, we can characterize the $s$-$t$ min cut on $H$.
\begin{lemma}\label{lem:goldberg-mincut}
    The $s$-$t$ min cut on $H$ has source side $S=\{s\}\cup S_V\cup E[S_V]$, where
    $S_V = \arg\max_{X\subseteq V} (c(E[X])-\tau|X|)$.
\end{lemma}
\begin{proof}
     For any $s$-$t$ cut on $H$ with source side $S\ni s$, denote $S_E=S\cap E$ and $S_V=S\cap V$.
     From \Cref{fact:goldberg-cut-value}, we have the following:
     For any fixed $S_V\subseteq V$, starting from $S_E=\emptyset$, the cut value $d^+_H(S)$ is decreased if we add an edge in $E[S_V]$ to $S_E$, and the cut value becomes infinity if we add an edge outside $E[S_V]$ to $S_E$. So, for fixed $S_V\subseteq V$, $d^+_H(S)$ is minimized when $S_E=E[S_V]$.
     So, the min cut value is $\min_{X\subseteq V}d^+_H(\{s\}\cup X\cup E[X])$.
     Because $d^+_H(\{s\}\cup X\cup E[X]) = c(E)-c(E[X])+\tau|X|$ by \Cref{fact:goldberg-cut-value} and $c(E)$ is a constant, the cut value is minimized when $c(E[X])-\tau|X|$ is maximized.
\end{proof}

Sometimes, we need to force some node $u\in V$ to be in the source side of the $s$-$t$ min cut. This can be achieved by adding an infinite capacity arc from $s$ to $u$.
\begin{lemma}\label{lem:goldberg-mincut-rooted}
    Construct $H'$ by adding to $H$ an infinite capacity arc from $s$ to $u\in V$.
    The $s$-$t$ min cut on $H'$ has source side $S=\{s\}\cup S_V\cup E[S_V]$, where
    $S_V = \arg\max_{u\in X\subseteq V} (c(E[X]) - \tau|X|)$.
\end{lemma}
\begin{proof}
For any $s$-$t$ cut on $H$ with source side $S\ni s$, denote $S_E=S\cap E$ and $S_V=S\cap V$.
When $u\in S_V$, $d^+_{H'}(S)=d^+_{H}(S)$ because the new arc $(s, u)$ does not contribute to $\partial^+S$. When $u\notin S_V$, $d^+_{H'}(S)=\infty$ due to the new arc.

We have shown in \Cref{lem:goldberg-mincut} that $d^+_H(S)$ is minimized when $S_E=E[S_V]$ for fixed $S_V$. So, $d^+_{H'}(S)$ is also minimized when $S_E=E[S_V]$ for fixed $S_V$, and $d^+_{H'}(\{s\}\cup X\cup E[X]) = c(E)-c(E[X])+\tau|X|$ when $u\in X$.
Since $c(E)$ is a constant, $d^+_{H'}(S)$ is minimized when $S_V=\arg\max_{u\in X\subseteq V} (c(E[X]) - \tau|X|)$.
\end{proof}

Next, we construct a modified network $\tilde{H}$ based on Goldberg's network $H$ (See \Cref{fig:goldberg}). The construction is due to \cite{Gabow95}.

We start from the residual network $H_R$ of the $s$-$t$ max flow on $H$.
We first remove $s$ and all its incident arcs from $H_R$.
Then, we shortcut all edge nodes as follows.
Consider any edge node $e=(u, v)\in E$. In $H_R$, $e$ is incident to 4 arcs, among which $(e, u)$ and $(e, v)$ are infinite edges from $H$, and $(u, e)$ and $(v, e)$ are created by the residual network, possibly of capacity $0$. We perform the following:
\begin{enumerate}
    \item Replace arcs $(u, e), (e, v)$ by $(u, v)$, and set capacity $c_{\tilde{H}}(u, v)=c_{H_R}(u, e)$.
    \item Replace arcs $(v, e), (e, u)$ by $(v, u)$, and set capacity $c_{\tilde{H}}(v, u)=c_{H_R}(v, e)$.
    \item Remove the node $e$.
\end{enumerate}
After shortcutting all edge nodes, we get a modified network $\tilde{H}$. 
\begin{fact}\label{fact:tilde-H-size}
    $|V[\tilde{H}]|=O(n),\, |E[\tilde{H}]|=O(m+n)$.
\end{fact}
\begin{lemma}\label{lem:goldberg-modified-cut}
Suppose the $s$-$t$ max flow on $H$ is $c(E)=\sum_{e\in E}c(e)$.
Then for any cut $X\subseteq V$ in $\tilde{H}$,
$$d^+_{\tilde{H}}(X) = \tau|X|-c(E[X]).$$
\end{lemma}
\begin{proof}
Consider any $X\subseteq V$. Let $S=\{s\}\cup X\cup E[X]$. By \Cref{fact:goldberg-cut-value},
$d^+_H(S) = c(E) - c(E[X])+\tau|X|$.
$S$ is an $s$-$t$ cut, so its cut value in the residual network is, by \Cref{fact:residual},
\[d^+_{H_R}(S) = d^+_H(S) - c(E) = \tau|X| - c(E[X]).\]

Note that $d_H^+(s)=c(E)$ by construction of $H$. So, the assumption implies that all $(s, e)$ edges are saturated in the flow. Then, $(s, e)$ edges have no capacity in $H_R$, and removing $s$ from $S$ does not affect $d^+_{H_R}(S)$.

Next, we consider the effect of shortcutting $e=(u, v)\in E$. If $S\supseteq\{u, v\}$ or $S\cap \{u, v\}=\emptyset$, then $\partial^+_{H_R}(S)$ is not affected. Otherwise, assume $S\cap \{u, v\}=\{u\}$ (the case of $S\cap \{u, v\}=\{v\}$ is symmetric). Then $e$ is not in $E[X]$. Among the 4 arcs incident to $e$, only $(u, e)$ is in $\partial^+_{H_R}(S)$. After shortcutting, only $(u, v)$ is in $\partial^+_{\tilde{H}}(S)$. So, the cut value does not change. In conclusion $d^+_{\tilde{H}}(X) = d^+_{H_R}(S) = \tau|X|-c(E[X])$.
\end{proof}

\subsection{The {\sc Find-Star} Subroutine}\label{sec:find-star}
In this section, we design an algorithm that outputs some vertex set given $G$ and $k$ in $\tO(mk\log C+F(m, m)\log C)$ time. Moreover, if $k$ is larger than the size of maximum skew-densest set, then the algorithm must output the maximum skew-densest set.

The algorithm needs the following subroutines. They are adapted from the global directed min cut algorithm \cite{CLNPQS21}. We include a proof of \Cref{lem:alg-sparsify} in \Cref{sec:partial-sparsify} for completeness.

\begin{lemma}[Lemmas 2.6 and 2.7 of \cite{CLNPQS21}]\label{lem:size-bounded-mincut}
    There exists an algorithm that, given a capacitated digraph $G$ with a fixed root vertex $t$, finds a $t$-mincut w.h.p.\ when the source side of a $t$-mincut has size $\le k$. The algorithm runs in $\tO(mk)$ time plus polylogarithmic calls to max flow on digraphs with $O(m)$ edges and $O(n)$ vertices and integer edge weights at most $\poly(n, C)$.
\end{lemma}
\begin{lemma}[Minimum Directed Cut]\label{lem:alg-sparsify}
    There exists an algorithm that,  given a capacitated digraph $G$ with a fixed root vertex $t$ and a number $\tau > 0$, assuming there exists a $t$-cut $S$ with $d^+(S)< \tau$ and $|S|\le k$, finds a $t$-cut $S$ with $d^+(S)< \tau$ w.h.p.
    The algorithm runs in $\tO(mk)$ time plus polylogarithmic calls to max flow on digraphs with $O(m)$ edges and $O(n)$ vertices and integer edge weights at most $\poly(n, C)$.
\end{lemma}

The algorithm performs a binary search for $\tau$ up to a precision of $n^{-3}$. Let $\tau^*$ be the maximum $\tau$ (in binary search) that the following procedure succeeds. The algorithm outputs the $t$-mincut of $\tilde{H}(G, \tau^*)$ using \Cref{lem:size-bounded-mincut}.
\begin{enumerate}
    \item Construct Goldberg's network $H$ from $G$ with parameter $\tau$.
    \item Compute the $s$-$t$ max flow on $H$. If the flow value $< c(E)$, then return success.
    \item Construct the modified network $\tilde{H}$.
    \item Call \Cref{lem:alg-sparsify} on input $(\tilde{H}, t, \tau, k)$ to find a set $U$.
    \item If $d^+_{\tilde{H}}(U)< \tau$, return success. Otherwise, return failure.
\end{enumerate}

\begin{algorithm}
\caption{{\sc Find-Star}($G, k$)}
\label{alg:find}
$\tau_L \gets 0, \tau_R\gets c(E)$\\
\While{$\tau_R-\tau_L \ge n^{-3}$}{
$\tau\gets \frac{\tau_L+\tau_R}{2}$\\
Construct Goldberg's network $H$ from $G$ with parameter $\tau$\\
\eIf{the $s$-$t$ max flow value on $H$ $< c(E)$}{
$\tau_L = \tau$
}{
Construct the modified network $\tilde{H}$ from $H$\\
Call \Cref{lem:alg-sparsify} on input $(\tilde{H}, t, \tau, k)$ to find a set $U$\\
\eIf{$d^+_{\tilde{H}}(U)< \tau$}{$\tau_L = \tau$}{$\tau_R = \tau$}
}
}
Output the $t$-mincut in $\tilde{H}$ of parameter $\tau_L$ using \Cref{lem:size-bounded-mincut}
\end{algorithm}

\begin{lemma}\label{lem:success-iff-tau-less-rhoD}
Let $D$ be the maximum skew-densest subgraph. Suppose $k\ge |D|$. The following holds w.h.p.: The binary search procedure on input $\tau$ succeeds if and only if $\tau < \rho(D)$.
\end{lemma}
\begin{proof}
   (Necessity) Assume the algorithm succeeds. There are two cases: success at step 2 or step 5. In both cases, we show $\tau < \rho(U)$ for some $U$, which implies $\tau< \rho(D)$ because $D$ is the skew-densest set.

    Suppose the algorithm succeeds at step 2. Let $S$ be the source side of the $s$-$t$ min cut in $H$ computed at step 2. By \Cref{lem:goldberg-mincut} and \Cref{fact:goldberg-cut-value}, $S=\{s\}\cup U_1\cup E[U_1]$ for some $U_1\subseteq V$, and $d_H^+(S) = c(E)-c(E[U_1])+\tau|U_1|$. According to step 2, we have $d_H^+(S) < c(E)$ or equivalently $c(E[U_1]) -\tau|U_1| > 0$. Then $|U_1|\ge 2$ and we can apply (\ref{eq:density-def-mult}) to get
    \[c(E[U_1]) = \rho(U_1)(|U_1|-1) > \tau|U_1|>\tau(|U_1|-1)\]
    which implies $\tau < \rho(U_1)$.

    Suppose the algorithm succeeds at step 5. Since it didn't succeed at step 2, the $s$-$t$ max flow on $H$ is $c(E)$. Let $U_2$ be the set found at step 4. By \Cref{lem:goldberg-modified-cut}, $d^+_{\tilde{H}}(U_2) = \tau|U_2|-c(E[U_2])$. According to step 5, we have $d^+_{\tilde{H}}(U_2) < \tau$. Then $|U_2|\ge 2$ and we can apply (\ref{eq:density-def-mult}) to get
    \[\tau(|U_2|-1) < c(E[U_2]) = \rho(U_2)(|U_2|-1)\]
     which implies $\tau < \rho(U_2)$.

    (Sufficiency) Assume $\tau < \rho(D)$. We need to show that the algorithm succeeds w.h.p.

    If the algorithm succeeds at step 2, we are done. Next assume it didn't, which means the $s$-$t$ max flow on $H$ is $c(E)$. By \Cref{lem:goldberg-modified-cut}, \[d^+_{\tilde{H}}(D) = \tau|D|-c(E[D]) = \tau|D|-\rho(D)(|D|-1) = \tau + (\tau-\rho(D))(|D|-1)< \tau\]
    So there exists a $t$-cut $D$ in $\tilde{H}$ with value $< \tau$ and $|D|\le k$.
    By \Cref{lem:alg-sparsify}, step 4 finds a cut of value $< \tau$ w.h.p., and the algorithm will succeed at step 5.
\end{proof}

\begin{lemma}\label{lem:alg-find-densest-set}
    Suppose the input graph $G$ is nonempty and has integer edge capacities. Let $D$ be the maximum skew-densest subgraph of $G$. Suppose $\rho(D)-n^{-3} < \tau < \rho(D)$.  Let $H$ and $\tilde{H}$ be Goldberg's network and the modified network of $G$ with parameter $\tau$.
    Then, $D$ is the $t$-mincut of $\tilde{H}$.
\end{lemma}
\begin{proof}
    Let $S=\{s\}\cup U_1\cup E[U_1]$ be the $s$-$t$ min cut in $H$.
    We claim that $U_1=\emptyset$. Assume for contradiction that $U_1$ is nonempty. Then by (\ref{eq:density-def-mult}),
    \[c(E[U_1]) = \rho(U_1)(|U_1|-1) \le \rho(D)(|U_1|-1)\]
    Combined with \Cref{fact:goldberg-cut-value},
    \begin{align*}
        d^+_H(S) &= c(E) - c(E[U_1]) + \tau |U_1|\\
        &\ge c(E) - \rho(D)(|U_1|-1) + (\rho(D)-n^{-3})|U_1|\\
        &\ge c(E)+\rho(D) - n^{-2}
    \end{align*}
    Since edge capacities are integral, the total edge capacity in $G$ is $\ge 1$, so $\rho(D)\ge \rho(V) \ge \frac{1}{n} > n^{-2}$.
    But $\partial^+_H(\{s\})$ has smaller cut value $c(E)$, which contradicts the assumption that $S$ is a $s$-$t$ min cut. We conclude that the $s$-$t$ mincut in $H$ is $\{s\}$, which has value $c(E)$.

    Since the $s$-$t$ max flow in $H$ also has value $c(E)$, we can apply \Cref{lem:goldberg-modified-cut} to obtain, for any nonempty $U\subseteq V$,
    \begin{equation}\label{eq:residual-cut-to-density}
        d^+_{\tilde{H}}(U) = \tau |U|-c(E[U])
    = \tau |U|-\rho(U)(|U|-1)=\tau + (\tau-\rho(U))(|U|-1)
    \end{equation}
    In particular, $d^+_{\tilde{H}}(D) < \tau$ because $\rho(D)>\tau$ and $|D|\ge 2$.

Recall that we need to show that $D$ is the $t$-mincut of $\tilde{H}$. For any set $U\subseteq V$ with $\rho(U) < \tau$ or $|U|=1$, we have $d^+_{\tilde{H}}(U) \ge \tau$ by (\ref{eq:residual-cut-to-density}), so $U$ is not the $t$-mincut.
    Next, we consider the remaining sets $U\subseteq V$ which satisfy $\rho(U)\ge \tau$ and $|U|\ge 2$. Because skew-density is in a set of fractional numbers with denominators in $[1, n-1]$, the gap between two different values of skew-density is at least $n^{-2}$.
    So, there cannot be another skew-density between $\rho(D)-n^{-3}$ and $\rho(D)$. It follows that $\rho(U)\ge \rho(D)$, and $U$ is a skew-densest set. Among the skew-densest sets $U$ with $\rho(U)=\rho(D)>\tau$, $d^+_{\tilde{H}}(U)$ decreases as $|U|$ increases by (\ref{eq:residual-cut-to-density}). So, $d^+_{\tilde{H}}(U)$ is minimized by the maximum skew-densest set, which is exactly $D$.
\end{proof}

We are now ready to prove \Cref{lem:alg-star}.

\FindStar*

\begin{proof}
    By \Cref{lem:success-iff-tau-less-rhoD}, the binary search succeeds if and only if $\tau < \rho(D)$. The algorithm finds the maximum successful $\tau^*$ up to a precision  of $n^{-3}$, so $\rho(D)-n^{-3} < \tau^* < \rho(D)$. By \Cref{lem:alg-find-densest-set}, for $\tilde{H}$ constructed from $(G, \tau^*)$, the $t$-mincut is $D$. The algorithm finds such a $t$-mincut by calling \Cref{lem:size-bounded-mincut}.

    The algorithm runs $O(\log (nC))$ iterations of binary search because the total edge capacity is at most $mC$ and the precision is $n^{-3}$. In each iteration, steps 1 and 3 take near-linear time to construct $H$ and $\tilde{H}$. Step 2 runs a max flow on $H$ with $O(m)$ vertices and edges, which takes $F(m, m)$ time. Step 4 takes $\tO(mk+F(n, m))$ time according to \Cref{lem:alg-sparsify}. Finally, after the binary search terminates, the algorithm calls  \Cref{lem:size-bounded-mincut} which takes $\tO(mk+F(n, m))$ time. In conclusion the running time of an iteration is $\tO(mk+F(m, m))$.
    The total running time is $\tO(mk\log C+F(m, m)\log C)$.
\end{proof}

\subsection{The {\sc Verify-Core} Subroutine}\label{sec:verify-star}
In this section, given $S\subseteq V$ and $k$ such that $|S|\le k$, we give an algorithm that can verify whether $S$ is a dense core in $O(k\cdot F(m, m))$ time.

The algorithm verifies whether $S$ is denser than its subsets using a network $\tilde{H}_1$ constructed from the induced subgraph $G[S]$, and verifies whether $S$ is denser than its supersets using a network $H_2'$ constructed from the contracted graph $G/S$.

\begin{algorithm}
\caption{{\sc Verify-Core}($G, k, S$)}
\label{alg:verify}
\If{$|S| > k$}{\Return False}
Construct $H_1$ to be Goldberg's network on $G[S]$ with parameter $\tau = \rho(S)$.\\
Construct $\tilde{H}_1$ to be the modified network of $H_1$.\\
Construct $H_2$ to be Goldberg's network on $G/S$ with parameter $\tau = \rho(S)$.\\
Construct $H_2'$ by add an infinite-capacity edge from $s$ to (contracted) $S$ on $H_2$.\\
\Return (the $s$-$t$ max flow value in $H_1$ $= c(E[S])$) ~ AND ~ (the $t$-mincut value in $\tilde{H_1}$ $ \ge \rho(S)$) ~ AND ~ (the $s$-$t$ max flow value in $H_2'$ $> |E[V/S]| + \rho(S)$)
\end{algorithm}

\begin{lemma}\label{lem:verify-correct-subsets}
    Let $H_1$ be Goldberg's network of $G[S]$ and $\rho(S)$. Let $\tilde{H}_1$ is the modified network of $H_1$.
    Then, $\rho(W) \le \rho(S)$ for all $W\subseteq S$ if and only if the $s$-$t$ max flow value in $H_1$ is $c(E[S])$, and the $t$-mincut in $\tilde{H_1}$ is at least $\rho(S)$.
\end{lemma}
\begin{proof}
   Let $\tau=\rho(S)$. The condition for $W=\emptyset$ is trivial, so we only consider nonempty $W$. Note that $\rho(W)\le \rho(S) \iff c(E[W]) \le \tau(|W|-1)$, so the condition is equivalent to
   \begin{equation}\label{eq:denser-than-subsets}
       \forall W\subseteq S\text{ s.t.\ }W\ne\emptyset,\, c(E[W]) \le \tau(|W|-1)
   \end{equation}

    (Sufficiency) Assume the $s$-$t$ max flow value in $H_1$ is $c(E[S])$ and the $t$-mincut in $\tilde{H}_1$ is $\ge \tau$. By \Cref{lem:goldberg-modified-cut}, for any nonempty $W\subseteq S$,
    \[d^+_{\tilde{H}_1}(W) =  \tau|W|-c(E[W])\ge \tau\]
    which implies (\ref{eq:denser-than-subsets}).

    (Necessity) Assume (\ref{eq:denser-than-subsets}) holds.
    By \Cref{fact:goldberg-cut-value}, for any nonempty $W\subseteq S$,
    \[d^+_{H_1}(\{s\}\cup W\cup E[W]) =  c(E)-c(E[W]) + \tau|W|
    \ge c(E)+\tau\]
     where the last inequality uses (\ref{eq:denser-than-subsets}). Then by \Cref{lem:goldberg-mincut}, the $s$-$t$ min cut in $H_1$ is $\partial^+\{s\}$, which has value $c(E[S])$.



    Now we can apply \Cref{lem:goldberg-modified-cut}. For any nonempty $W\subseteq S$,
    \[d^+_{\tilde{H}_1}(W) =  \tau|W|-c(E[W])\ge \tau\]
     where the last inequality uses (\ref{eq:denser-than-subsets}). So, the $t$-mincut in $\tilde{H}_1$ is $\ge \tau$.

\end{proof}

\begin{lemma}\label{lem:verify-correct-supsets}
    Let $H_2$ be Goldberg's network of $G[V/S]$ and $\rho(S)$. Construct $H_2'$ by adding an infinite-capacity edge from $s$ to (the contracted vertex) $S$ in $H_2$.
    Then, $\rho(U) < \rho(S)$ for all $U\supsetneqq S$ if and only if the $s$-$t$ max flow value in $H_2'$ is $> |E[V/S]| + \rho(S)$.
\end{lemma}
\begin{proof}
We first show that for $U\supsetneqq S$, $\rho(U)<\rho(S)\iff |E[U/S]| < \rho(S)(|U|-|S|)$.

Assume $\rho(U)<\rho(S)$. Then $\rho(S)>0$ and $|U|>|S|\ge 2$. By (\ref{eq:density-def-mult}),
\[|E[U]|=\rho(U)(|U|-1) < \rho(S)(|U|-1)\]
\[|E[U/S]| = |E[U]|-|E[S]| < \rho(S)(|U|-1)-\rho(S)(|S|-1) = \rho(S)(|U|-|S|)\]

Assume $|E[U/S]| < \rho(S)(|U|-|S|)$. Then $\rho(S)>0$ and $|U|>|S|\ge 2$. By (\ref{eq:density-def-mult}),
\[|E[U]|=|E[S]|+|E[U/S]| < \rho(S)(|S|-1)+\rho(S)(|U|-|S|) = \rho(S)(|U|-1)\]
Combined with $|E[U]|=\rho(U)(|U|-1)$, we have $\rho(U)<\rho(S)$.

Let $\tau=\rho(S)$. Now, the condition is equivalent to
\begin{equation}\label{eq:denser-than-supsets}
    \forall U\supsetneqq S,\, |E[U/S]| < \tau(|U|-|S|)
\end{equation}

(Sufficiency) Assume the min cut value in $H_2'$ is $> |E[V/S]| + \tau$. For any $U\supsetneqq S$, let $X=U/S$. By \Cref{fact:goldberg-cut-value},
\[d^+_{H_2'}(\{s\}\cup X\cup E[X]) = |E[V/S]|-|E[X]| + \tau|X| = |E[V/S]| -|E[U/S]|+\tau(|U|-|S|+1)\]
Combined with the assumption that $d^+_{H_2'}(\{s\}\cup X\cup E[X]) \ge($min cut value in $H_2') > |E[V/S]| + \tau$, we have $|E[U/S]| < \tau(|U|-|S|)$.

(Necessity) Assume (\ref{eq:denser-than-supsets}) holds.
By \Cref{lem:goldberg-mincut-rooted}, the $s$-$t$ min cut on $H_2'$ has source side $\{s\}\cup X\cup E[X]$ for some $X$ that contains the contracted $S$. Expand the contracted $S$ in $X$ to form $U$.

By \Cref{fact:goldberg-cut-value}, the min cut value is
\[|E[V/S]|-|E[X]| + \tau|X| = |E[V/S]| -|E[U/S]|+\tau(|U|-|S|+1) > |E[V/S]| + \tau\]
where the last inequality uses (\ref{eq:denser-than-supsets}).
\end{proof}

We are now ready to prove \Cref{lem:alg-verify}.

\VerifyCore*

\begin{proof}
    The correctness follows from \Cref{lem:verify-correct-subsets,lem:verify-correct-supsets}.

    Next, we analyze the running time. Constructing $H_1, H_2$ and $H_2'$ takes $O(m)$ time. Constructing $\tilde{H}_1$ calls a max flow on $H_1$ and takes $O(F(m, m))$ time. Computing $s$-$t$ max flow in $H_1$ and $H_2$ takes $O(F(m, m))$ time. Finally, to compute $t$-mincut in $\tilde{H}_1$, we iterate over all $s\in S$, and take the minimum among all $s$-$t$ mincuts. Because $|S|\le k$, the running time is $O(k\cdot F(n, m))$. In conclusion the total running time is $O(k\cdot F(m, m))$.
\end{proof}

%% file: arboricity.tex
\newcommand{\tGamma}{\tilde{\Gamma}}

In this section, we give a faster algorithm for computing the arboricity $\Gamma(G)$ of a graph $G$, which builds on ideas presented in the previous section. Nash-Williams~\cite{nash-williams} showed  that arboricity is equal to the integer ceiling of the maximum skew-density of any induced subgraph in $G$. We denote this maximum skew-density $\tGamma(G)$, i.e., $\Gamma(G) = \lceil \tGamma(G)\rceil$.
We will show that a variant of \Cref{alg:find} computes $\Gamma(G)$ in $\tO(1)$ calls to two subroutines -- directed global min-cut and max-flow -- and $\tO(m)$ additional time outside these oracle calls. We restate the formal theorem below.



\arboricityAlg*

In the directed global min-cut problem, the goal is to find a set of vertices $S\subset V$ that minimizes the value of the in-cut, i.e., the sum of capacities of directed edges from $V\setminus S$ to $S$. Instead of reducing arboricity to the directed global min-cut problem, it is easier for exposition to reduce arboricity to the related problem of finding a $\bar{t}$-minimum cut. Formally, a $\bar{t}$-minimum cut in a directed graph $G=(V\cup\{t\},E,c)$ is a subset $S\subset V$ that minimizes the sum of edge capacities going from $V\cup\{t\}\setminus S$ to $S$. That is, a $\bar{t}$-minimum cut is a minimum cut in the graph under the restriction that the vertex subset defining the cut cannot contain $t$. The next lemma makes the simple observation that this is essentially equivalent to the directed global min-cut problem:
\begin{fact}\label{fact:t-mincut-to-mincut}
    Finding a $\bar{t}$-minimum cut requires one call to a directed global min-cut oracle. Conversely, finding a directed global min-cut requires two calls to a $\bar{t}$-minimum cut oracle.
\end{fact}
\begin{proof}
    To find a $\bar{t}$-minimum cut, take the graph $G=(V\cup\{t\},E,c)$ and add edges $(v,t)$ for each $v\in V$ with capacities $c(v,t)=\infty$. (We can emulate an edge of infinite capacity by setting its capacity larger than the total capacity of the graph.) Then any minimum (in-)cut $S$ must not include the vertex $t$, i.e., $S\subseteq V$, which corresponds exactly to the $\bar{t}$-minimum cut. To find a directed global min cut in $G=(V,E,c)$, choose an arbitrary vertex $t\in V$ and compute the $\bar{t}$-minimum cut in $G$ and in the graph where every edge direction is reversed, and return the smaller of these two cuts. Since $t$ is either on the side of the minimum cut or it is not, this finds the global minimum cut.
\end{proof} 

We now give our algorithm for computing arboricity in \Cref{alg:arboricity}. There is an outer binary search procedure, where in each step, the guessed value of arboricity is denoted $\tau$. To test if arboricity exceeds $\tau$ or is smaller than it, we use two subroutine calls -- maximum flow on Goldberg's network with parameter $\tau$ and $\bar{t}$-minimum cut on the modified version of Goldberg's network from the previous section.

\begin{algorithm}
\caption{{\sc Compute-Arboricity}($G$)}
\label{alg:arboricity}
$\tau_L \gets 0, \tau_R\gets c(E)$\\
\While{$\tau_R-\tau_L \ge n^{-3}$}{
$\tau\gets \frac{\tau_L+\tau_R}{2}$\\
Construct Goldberg's network $H$ from $G$ with parameter $\tau$\label{line:alg-arboricity-build1}\\
\eIf{the $s$-$t$ max flow value on $H$ $< c(E)$\label{line:alg-arboricity-cond1}}{
$\tau_L = \tau$ \label{line:alg-arboricity-case1}
}{
Construct the modified network $\tilde{H}$ from $H$\label{line:alg-arboricity-build2}\\
\eIf{$\bar{t}$-minimum cut of $\tilde{H}<\tau$\label{line:alg-arboricity-cond2}}{$\tau_L = \tau$\label{line:alg-arboricity-case2}
}{$\tau_R = \tau$\label{line:alg-arboricity-case3}
}
}
}
Output $\lceil \tau_L\rceil$.
\end{algorithm} 

In the next lemma, we argue correctness of a single iteration of the binary search procedure in \Cref{alg:arboricity}:

\begin{lemma}\label{lem:success}
    The binary search procedure on input $\tau$ sets $\tau_L=\tau$ if and only if $\tau < \tGamma(G)$. 
\end{lemma}
\begin{proof}
    If the algorithm sets $\tau_L=\tau$, there are two cases: Line \ref{line:alg-arboricity-case1} and Line \ref{line:alg-arboricity-case2}. In both cases, we show $\tau < \rho(U)$ for some $U$, which implies $\tau< \tGamma(G)$ because $\tGamma(G)$ is the value of the skew-densest set.

    Suppose the algorithm sets $\tau_L=\tau$ on Line \ref{line:alg-arboricity-case1}. Let $S$ be the source side of the $s$-$t$ min cut in $H$. By \Cref{lem:goldberg-mincut} and \Cref{fact:goldberg-cut-value}, $S=\{s\}\cup U_1\cup E[U_1]$ for some $U_1\subseteq V$, and $d_H^+(S) = c(E)-c(E[U_1])+\tau|U_1|$. Since the \textbf{if} statement on Line \ref{line:alg-arboricity-cond1} succeeds, we have $d_H^+(S) < c(E)$ or equivalently $c(E[U_1]) -\tau|U_1| > 0$. Then we can apply (\ref{eq:density-def-mult}) to get
    \[\rho(U_1)(|U_1|-1) = c(E[U_1]) > \tau|U_1|>\tau(|U_1|-1)\]
    which implies $\tau < \rho(U_1)$.

    Suppose the algorithm sets $\tau_L=\tau$ in Line \ref{line:alg-arboricity-case2}. Since this didn't happen at Line \ref{line:alg-arboricity-case1}, the $s$-$t$ max flow on $H$ is $c(E)$. Let $U_2$ be the $\bar{t}$-minimum cut of $\tilde{H}$ found on Line \ref{line:alg-arboricity-cond2}. By \Cref{lem:goldberg-modified-cut}, $d^+_{\tilde{H}}(U_2) = \tau|U_2|-c(E[U_2])$. Since the \textbf{if} statement on Line \ref{line:alg-arboricity-cond2} succeeds, we have $d^+_{\tilde{H}}(U_2) < \tau$. Then we can apply (\ref{eq:density-def-mult}) to get
    \[\tau(|U_2|-1) < c(E[U_2]) = \rho(U_2)(|U_2|-1)\]
     which implies $\tau < \rho(U_2)$.

    Next, assume $\tau < \tGamma(G)$. We need to show that the algorithm sets $\tau_L=\tau$. If the algorithm does so at Line \ref{line:alg-arboricity-case1}, we are done. Assume it didn't, which means the $s$-$t$ max flow on $H$ is $c(E)$. Let $G[D]$ be the induced subgraph of maximum skew-density ($|D| \ge 2$). By Nash-Williams~\cite{nash-williams}, we have $\rho(D) = \tGamma(G) > \tau$. By \Cref{lem:goldberg-modified-cut}, \[d^+_{\tilde{H}}(D) = \tau|D|-c(E[D]) = \tau|D|-\rho(D)(|D|-1) = \tau + (\tau-\rho(D))(|D|-1)< \tau.\]
    Thus, there exists a $\bar{t}$-cut $D$ in $\tilde{H}$ with value $< \tau$, so the minimum $\bar{t}$-cut value of $\tilde{H}$ will be found to be $<\tau$ in Line \ref{line:alg-arboricity-cond2}, as desired.
\end{proof}

We now conclude the proof of \Cref{thm:arboricity-intro}:

\begin{proof}[Proof of \Cref{thm:arboricity-intro}.]
    By \Cref{lem:success}, each iteration of the binary search sets $\tau_L=\tau$ if and only if $\tau < \tGamma(G)$. Let the values of $\tau_L$ and $\tau_R$ in the final iteration of binary search be denoted $\tau_L^*$ and $\tau_R^*$; then, $\tau_R^* -\tau_L^* < n^{-3}$ and $\tGamma(G)\in [\tau_L^*, \tau_R^*]$. Note that $\tGamma(G)$ can be written as a fraction of two integers with denominator in the range $\{1, \ldots, n-1\}$. Any two distinct rational numbers that can be written in this form must differ by at least $n^{-2}$. It follows there is only one candidate value of $\tGamma(G)$ in the range $[\tau_L^*, \tau_R^*]$. Moroever, every integer is a candidate. Now, there are two cases. If $[\tau_L^*, \tau_R^*]$ contains an integer, then $\tGamma(G)$ must be equal to this integer since the latter is the unique candidate in this range. Therefore, $\Gamma(G) = \tGamma(G) = \lceil \tau_L^*\rceil$. In the other case, there is no integer in the range $[\tau_L^*, \tau_R^*]$. In this case, since $\tGamma(G)\in [\tau_L^*, \tau_R^*]$, we have that $\Gamma(G) = \lceil \tGamma(G) \rceil = \lceil \tau_L^* \rceil$. Therefore, $\Gamma(G) = \lceil \tau_L^*\rceil$ in both cases.

    Now, we bound the running time of \Cref{alg:arboricity}.
    Recall that we denote $F(m,n)$ to be the running time of max flow on a directed graph with $m$ edges and $n$ vertices. Similarly, denote $M(m,n)$ to be the running time of global min cut on a directed graph with $m$ edges and $n$ vertices.
    The algorithm runs $O(\log (nC))$ iterations of binary search because the total edge capacity is at most $n^2 C$ and the precision is $n^{-3}$. In each iteration, Lines \ref{line:alg-arboricity-build1} and \ref{line:alg-arboricity-build2} take $O(m  + F(m,n))$ time to construct $H$ and $\tilde{H}$, as shown in the proof of \Cref{lem:alg-verify}. Line \ref{line:alg-arboricity-cond1} runs a max flow on $H$ with $O(m)$ vertices and $O(m)$ edges by \Cref{fact:H-size}. Line \ref{line:alg-arboricity-cond2} can be computed by two calls to directed min-cut on graphs with $m$ edges and $n$ vertices according to \Cref{fact:t-mincut-to-mincut,fact:tilde-H-size}. In conclusion the running time of an iteration is $O(m + M(m,n)+F(m, m))$. The total running time is $O((m + M(m,n) + F(m, m))\log (nC))$.
\end{proof}

%% file: entropy.tex
In this section, we describe an application of the cut hierarchy to an ideal tree packing, and in turn, to the max-entropy fractional spanning tree.

\subsection{Ideal Load and Ideal Tree Packing}
Given a connected undirected graph $G$ with edge weights $c_e$, Thorup~\cite{Thorup01, Thorup08} defines the \emph{ideal load} as a function on edges. We restate the definition below based on the cut hierarchy.

For each node $p$ in the cut hierarchy, the node is associated with a vertex set $S_p$. Let $G_p = (V_p, E_p)$ be the graph formed from the induced subgraph $G[S_p]$ by contracting all vertex sets associated with children of $p$ in the cut hierarchy. By the definition of the cut hierarchy, we have that the min-ratio cut in $G_p$ is the all-singleton cut.
For each edge $e\in E$, we associate $e$ with the deepest node in the cut hierarchy whose induced subgraph contains $e$; we denote this node $p(e)$.

The ideal load is defined by
$$\ell(e) = \frac{c_e}{\sigma(G_{p(e)})} = \frac{c_e}{c(E_{p(e)}) / (|V_{p(e)}|-1)}.$$

The following lemma shows that the ideal load can be computed from the cut hierarchy in $O(m)$ time.
\begin{lemma}\label{lem:cut-hierarchy-to-ideal-load}
Given the cut hierarchy of an undirected graph $G$, the ideal load can be computed in $O(m)$ time.
\end{lemma}
\begin{proof}
   For each edge $e=(u, v)$, $p(e)$ is the LCA of $u$ and $v$ (as leaves of the cut hierarchy) in the cut hierarchy. We can use Tarjan's offline LCA to compute $p(e)$ for all $e\in E$ in $O(m)$ time.

    Note that $\sigma(G_p)$ is the ratio of total weight of edges in $G_p$ over the number of children of $p$ minus 1. By scanning all edges and accumulate $c_e$ to $p(e)$, we can obtain the total weight of edges in $G_p$ for all nodes $p$. In conclusion all $\sigma(G_p)$ and hence all $\ell(e)$ can be computed in $O(m)$ time.
\end{proof}

Importantly, the ideal loads on the edges define a fractional spanning tree~\cite{Thorup01}. So, there exists a convex combination of spanning trees whose average matches the ideal loads. We show that we can efficiently sample from such a distribution using the cut hierarchy.

\begin{lemma}
    Given the cut hierarchy of a connected undirected graph $G$, we can construct a data structure in $\tO(n^3 m)$ time and $O(nm)$ space from which we can sample in $O(n)$ time a spanning tree that belongs to a fixed ideal tree packing.
\end{lemma}
\begin{proof}
    For each node $p$ of the cut hierarchy, we construct a maximal fractional tree packing $\mathbb{T}_p$ in $G_p$.
    (By \Cref{fact:min-ratio-cut-connected}, each $G_p$ is connected, so $\mathbb{T}_p$ is a linear combination of spanning trees in $G_p$.)
    Let $n_p, m_p$ be the number of vertices and edges in $G_p$ respectively.
    The construction time is $\tO(n_{p}^3m_{p})$ for each $G_p$ by using the fractional tree packing algorithm of Gabow and Manu~\cite{GabowManu95}. Since $n_p$ equals the number of children of node $p$ in the cut hierarchy, and the total number of nodes in the cut hierarchy is $\le 2n$ since they represent a laminar family of subsets of $V$, we can conclude that $\sum_p n_{p}\le 2n$. Similarly, since every edge appears in only one of the graphs $G_p$, we can conclude that $\sum_p m_{p}\le m$. So, the total construction time is $\tO(n^3 m)$.
    Gabow-Manu algorithm guarantees that each $\mathbb{T}_p$ consists of at most $m_p$ distinct trees, so the space is at most $\sum_p m_{p} n_{p} = O(nm)$.
    
    To draw a sample, we simply draw a random tree from each $\mathbb{T}_p$ (each tree is drawn with probability proportional its weight in the packing), and take the union (in terms of the edges) of all trees drawn. We claim that the resulting graph is a spanning tree. First, by bottom-up induction on the cut hierarchy, we have that each $S_p$ (the vertex set corresponding to node $p$) is connected by the sample. As a result, the sample spans the whole graph. Second, the number of edges in the union is the sum over $p$ of rank difference of $S_p$ and all children of $p$. The sum telescopes to be the rank difference of whole graph and all singletons, which is $n-1$. Because the sample has $n-1$ edges and is spanning, it must be a spanning tree.

    We now compute the probability that ane edge $e$ is chosen in the sample.
    Notice that $e$ only appears in one of the graphs $G_p$, namely $G_{p(e)}$. So, we only need to consider the packing $\mathbb{T}_{p(e)}$, which has value $\sigma(G_{p(e)})$.
    Because the min-ratio cut of $G_{p(e)}$ is the all-singleton cut, all edges are fully used in the maximum fractional tree packing, i.e., $e$ is used in a subset of trees with total value $c_e$. So, the probability that $e$ is chosen in $\mathbb{T}_{p(e)}$ is $c_e/\sigma(G_{p(e)})=\ell(e)$, as desired.
\end{proof}

\subsection{Ideal Load as Maximum Entropy Fractional Tree}
In this section, we show that the fractional spanning tree defined by the ideal loads maximizes the entropy function in the spanning tree polytope. This could be a natural alternate definition of ideal load.

To validate this observation, we view each integer-weighted edge $e$ as $c_e$ parallel, unweighted edges. Then, each unweighted edge has ideal load $\ell(e) = 1/\sigma(G_{p(e)})$.
Denote $x_e$ to be the value of edge $e$ in the spanning tree polytope.
The Shannon entropy of the (normalized) edge values is $H(x) := \sum_e \frac{x_e}{n-1}\ln \frac{n-1}{x_e} = \ln (n-1) - \frac{1}{n-1}\sum_e x_e\ln x_e$. So, maximizing entropy is equivalent to minimizing $\sum_e x_e \ln x_e$ in the spanning tree polytope.

\begin{theorem}\label{thm:entropy}
    In an unweighted multigraph, the ideal load is the maximizer of entropy of the edge values in the spanning tree polytope.
\end{theorem}
We prove this theorem in the rest of this section.

We want to maximize entropy given by $\min \sum_{e\in E}x_e\ln x_e$ subject to the constraints of the spanning tree polytope:
\[
    x_e\in [0, 1] ~\forall e\in E \text{ such that } x(E[S]) := \sum_{e\in E[S]} x_e \le |S|-1 ~\forall S\subseteq V \text{ and } \sum_{e\in E} x_e = n-1.
\]
The first group of constraints can be restricted to $2\le |S|\le n-1$, because $|S|=1$ is trivial and $|S|=n$ is included in the normalization constraint. 
Note also that the constraints $x_e\le 1$ are automatically implied by the subset constraints for pairs of vertices, so we only need to explicitly enforce $x_e \ge 0$ for all $e\in E$.

We introduce Lagrangian multipliers $\mu_S$ for $S\subseteq V, |S|\ge 2$.
The Lagrangian is
\[L(x, \mu) = \sum_e x_e\ln x_e + \sum_{|S|\ge 2}\mu_S \cdot (x(E[S]) - |S| + 1)\]
The primal program is equivalent to $\max_\mu \min_x L(x, \mu)$  where $\mu_S\ge 0$ for sets with $2\le |S|\le n-1$, and $\mu_V$ is unconstrained.  Next we solve $g(\mu):=\min_{x\ge 0} L(x, \mu)$ to obtain the dual program $\max_\mu g(\mu)$. Let $x^*$ denote $\arg\min_{x\ge 0} L(x, \mu)$. Taking the partial derivatives at $x = x^*$, we get that for every edge $e$,
\[
1 + \ln x^*_e + \sum_{S: e\in E[S]} \mu_S = 0,
\]
which implies that 
\[
x^*_e = \exp\left(-1-\sum_{S: e\in E[S]} \mu_S\right).
\]
\eat{
Substituting in the Lagrangian $L$, we get the dual objective
\begin{align*}
g(\mu) &= \sum_e x^*_e \ln x^*_e + \sum_S \mu_S \sum_{e\in E[S]}x^*_e - \sum_S \mu_S\cdot (|S|-1)\\
&= -\sum_e \exp\left(-1-\sum_{S: e\in E[S]} \mu_S\right)- \sum_S \mu_S\cdot (|S|-1)
\end{align*}
}
Define 
\[
    y_S = \begin{cases}
    \mu_S & S\ne V \\
    \mu_S+1 & S=V,
\end{cases}
\]
so that $x^*_e = \exp(-\sum_{S:e\in E[S]} y_S)$.
Substituting in the Lagrangian $L$, we get the dual objective
\begin{align*}
&= \sum_e x^*_e \cdot \left(-\sum_{S:e\in E[S]} y_S\right) + \sum_S y_S\cdot \left(\sum_{e\in E[S]}x^*_e - |S| + 1\right)  - (x^*(E)-n+1)\\
&= \sum_S y_S\cdot (-|S|+1) - (x^*(E)-n+1)\\
&= n-1-\sum_e \exp\left(-\sum_{S:e\in E[S]}y_S\right) - \sum_S y_S(|S|-1)
\end{align*}
The dual program is 
\[
    \max~ g(y) := n-1-\sum_e \exp\left(-\sum_{S:e\in E[S]}y_S\right) - \sum_S y_S(|S|-1)
    \text{ subject to }~ y_S\ge 0 ~\forall S\ne V.
\]

Because the ideal load is in the spanning tree polytope, it is a feasible primal solution.
We now construct a dual solution $y^*$ and show its dual objective matches the primal objective for ideal load. It follows that both solutions are optimal.

For each non-root node $p$ in the cut hierarchy with parent $q$, set $y^*_{S_{p}} := \ln \sigma(G_p) - \ln \sigma(G_q)$, where $q$ is the parent of $p$ in the cut hierarchy. For the root node corresponding to $V$, set $y^*_V := \ln \sigma(G)$.
It is easy to verify the next property, which implies dual feasibility:
\begin{fact}\label{fact:delta-positive}
    Consider any non-root node $p$ of the cut hierarchy and let $q$ denote its parent. Then, $\sigma(G_p) \ge \sigma(G_q)$, which implies $y^*_{S_p} \ge 0$.
\end{fact}
\begin{proof}
 Consider any non-root node $p$ and let $q$ be $p$'s parent. Assume for contradiction that $\sigma(G_p) < \sigma(G_q)$. 
    Let $\mathcal{P}, \mathcal{Q}$ be the maximal min-ratio cuts in $G[S_p], G[S_q]$ respectively. Note that $S_p$ is a side in $\mathcal{Q}$. Let $\mathcal{Q}'=(\mathcal{Q}\setminus \{S_p\}) \cup \mathcal{P}$.
    The cut ratio of $\mathcal{Q}'$ is
    \[\frac{d_{G[S_q]}(\mathcal{Q}')}{|\mathcal{Q}'|-1}
    = \frac{d_{G[S_q]}(\mathcal{Q})+d_{G[S_p]}(\mathcal{P})}{|\mathcal{Q}|-1+|\mathcal{P}|-1} 
     = \frac{\sigma(G_q)(|\mathcal{Q}|-1)+\sigma(G_p)(|\mathcal{P}|-1)}{|\mathcal{Q}|-1+|\mathcal{P}|-1} < \sigma(G_q)
    \]
    which contradicts the fact that $\mathcal{Q}$ is a min-ratio cut in $G[S_q]$.
\end{proof}

Because the sets form a laminar family, we have
$\sigma(G_p) = \exp\left(\sum_{q:S_q\supseteq S_p} y^*_{S_q}\right)$.
Therefore,
\begin{equation}\label{eq:primal-dual-soln-relation}
    \ell(e) = \frac{1}{\sigma(G_{p(e)})}
= \exp\left(-\sum_{S: e\in E[S]} y^*_S\right)
\end{equation}

\eat{
We now extend the recursive min-ratio cut process to obtain a dual solution. 
\paragraph{Primal-Dual min-ratio cut Process.}
Define a recursive procedure Dual-alg($S, \Delta$) as follows.
Find a min-ratio cut $\mathcal{P}$ of $G[S]$.
Set $y_S=\ln \frac{d_{G[S]}(\mathcal{P})}{|\mathcal{P}|-1}$ - $\Delta$.
Set $x_e =\frac{|\mathcal{P}|-1}{d(\mathcal{P})}$ for each $e\in \partial \mathcal{P}$.
Recursively call Dual-alg($U, \Delta+\mu_S$) for all sides of $\mathcal{P}$.

We call Dual-alg($V, 0$) to output the dual solution $y$.
$y$ is feasible by the following monotone property.

From the construction, we have $x_e=\exp(-1-\sum_{S\supseteq e}\mu_S)$ for each edge $e$.
$x$ is the same as definition of ideal tree packing, so it is a feasible primal solution.
}

Applying (\ref{eq:primal-dual-soln-relation}) to $g(y^*)$, we have
\[g(y^*) = n-1-\sum_e \ell(e) - \sum_S y^*_S(|S|-1) =  -\sum_S y^*_S(|S|-1).\]
Here, the last step uses $\sum_{e\in E}\ell(e)=n-1$ because $\ell$ is a fractional spanning tree.
On the other hand,
\[
\sum_e \ell(e)\ln \ell(e) 
= - \sum_e \ell(e) \sum_{S:e \in E[S]} y^*_S
= -\sum_S y_S^*\sum_{e\in E[S]}\ell(e)
= -\sum_S y_S^* \cdot \ell(E[S]).
\]
We can ignore the terms with $y^*_S=0$. The remaining $S$ with $y^*_S\ne 0$ are sets in the cut hierarchy.
It remains to show that for such sets, $\ell(E[S]) = |S|-1$.

Note that if we define the ideal load in a subgraph $G[S_p]$ for a node $p$ in the cut hierarchy, the definition is identical to $\ell(e)$ because the cut hierarchy in $G[S_p]$ is a sub-family of the cut hierarchy in $G$. So, $\ell$ is also a fractional spanning tree in $G[S_p]$, and we have $\sum_{e\subseteq S_p}\ell(e) = |S_p|-1$.
In conclusion $g(y^*)=\sum_e \ell(e) \ln \ell(e)$, which implies that $x^*_e=\ell(e)$ is an optimal solution for the primal program.

%% file: closing.tex
In this paper, we have presented new algorithms for constructing the cut hierarchy of a weighted, undirected graph in $nm^{1+o(1)}$ time, and for determining its arboricity in $\sqrt{n}m^{1+o(1)}$ time. These algorithms improve the state of the art for both problems, previously due to Gabow~\cite{Gabow95}, from $\tO(n^2 m)$ (weighted) and $\tO(n m^{3/2})$ (unweighted) for the cut hierarchy problem and $\tO(nm)$ (weighted) and $\tO(m^{3/2})$ (unweighted) for the arboricity problem. Our arboricity algorithm is bottlenecked by the running time of a directed minimum cut subroutine -- if that improves to $m^{1+o(1)}$, the entire algorithm will run in $m^{1+o(1)}$ time. Nevertheless, it would be interesting to explore whether this dependence on the directed minimum cut subroutine is necessary -- apriori, there is no reason to believe that it is. Improving the running time for the cut hierarchy problem beyond that shown in this paper is also an interesting goal. Among other implications, this would result in a faster algorithm for computing the strength of a graph beyond the current best running time of $\tO(nm)$.

%% file: sparsification.tex
In this section, we design an algorithm to solve the following subproblem:
Given a digraph $G$ with a root vertex $t$, cut value parameter $\tau$ and size parameter $k$, assuming there exists a $t$-cut $S$ with $d^+(S)< \tau$ and $|S|\le k$, find any $t$-cut $S'$ with $d^+(S')< \tau$.

The cut we find will be a minimum 1-respecting $t$-cut.
We define a $t$-arborescence to be a directed subgraph where each vertex except $t$ has out-degree $1$, $t$ has out-degree $0$, and each vertex has a directed path to $t$. We say a $t$-cut 1-respects an arborescence if the intersection of cut edges and arborescence edges is 1. 

\begin{lemma}\label{lem:1-resp-mincut-alg}
    There exists an algorithm that, given a capacitated digraph $G$ and a $t$-arborescence $T$, output the minimum $t$-cut that 1-respects $T$ in $\tO(m)$ time plus polylogarithmic calls to max flow on digraphs with $O(m)$ edges and $O(n)$ vertices and integer edge weights at most $\poly(n, W)$.
\end{lemma}
\begin{proof}
    The lemma is a corollary of the following result:
    \begin{lemma}[Theorem 2.7 of \cite{CLNPQS21}]\label{lem:1resp-cut-original}
    There exists an algorithm that, given a capacitated digraph $G$ and a $t$-arborescence $T$ that 1-respects a $t$-mincut, output a (not necessarily $1$-respecting) $t$-mincut in $\tO(m)$ time plus polylogarithmic calls to max flow.
    \end{lemma}
    (We remark that our direction of $t$-cut and $t$-arborescence is opposite to \cite{CLNPQS21}, but the results are equivalent after reversing orientation of all arcs.)
    
    The only difference to the lemma above is that we require the output $t$-mincut to be $1$-respecting. By adding a uniform large capacity on all edges of $T$, we can force the $t$-mincut to be 1-respecting, while keeping the relative value of all 1-respecting cuts. It follows that the original minimum 1-respecting cut is now the global $t$-mincut, and we can apply \Cref{lem:1resp-cut-original} to find it.
\end{proof}

It remains to construct a $t$-arborescence that 1-respects a cut of value $< \tau$. 
Our approach is adapted from \cite{CLNPQS21}.

\begin{lemma}\label{lem:partial-sparsify}
    Given a capacitated digraph $G$ with a fixed root vertex $t$ and a number $\tau > 0$, suppose there exists a $t$-cut $S$ with $d^+(S)< \tau$ and $|S|\le k$.
    In randomized nearly linear time, we can compute a sparsifier $G_0=(V_0, E_0)$, where $t\in V_0\subseteq V$ and the following properties hold.
    \begin{enumerate}
        \item $G_0$ has integer edge capacities and the $t$-mincut has value $O(k\log(n)/\eps^2)$.
        \item There exists some $S'$ with $d^+_G(S')< \tau$ such that $S'$ is a $(1+\eps)$-approximate $t$-mincut in $G_0$.
    \end{enumerate}
\end{lemma}
\begin{proof}
    Let $\mu=c_{\mu}\eps^2\tau/(k\log n)$ for a sufficiently small $c_{\mu}=\Theta(1)$ such that $\tau$ and $\eps\tau/2k$ are integer multiples of $\mu$ (we may decrease $\eps$ by a constant factor). Let $G_0$ be a capacitated graph obtained as follows.
    \begin{enumerate}
        \item Randomly round each edge capacity independently up or down to the nearest multiple of $\mu$, such that the expectation is $w_e$ (i.e., if $w_e=K\mu + r,\, K\in \mathbb{Z},\, r\in[0, \mu)$, sample $w'_e = K\mu$ w.p.\ $1-\frac{r}{\mu}$, and $w'_e = (K+1)\mu$ w.p.\ $\frac{r}{\mu}$).
        \item Add an edge of capacity $\eps\tau/2k$ from each $v\in V\setminus\{t\}$ to $t$.
        \item Scale down all the edge capacities by $\mu$.
    \end{enumerate}
    For each set $S\subseteq V$, let $d_1(S)$ be the original cut value, $d_2(S)$ be the cut value after step 1, and $d_3(S)$ be the cut value after step 2.

    Consider any fixed $S\subseteq V\setminus\{t\}$. $d_2(S)$ is the sum of $|\partial^+ S|$ independent random variables, and has expectation $d_1(S)$. Each random variable in the sum is nonnegative and varies by at most $\mu$. By a variation of Chernoff bound\footnote{\label{footnote:additive-chernoff} Here we
    apply the following bounds (appropriately rescaled) which follow
    from the same proof as the standard multiplicative Chernoff bound.
    \begin{quote}
      \itshape Let $X_1,\dots,X_n \in [0,1]$ independent random
      variables. Then for all $\eps > 0$ sufficiently small and all
      $\gamma > 0$,
      \begin{align*}
        \Pr[X_1 + \cdots + X_n \leq (1-\eps)\E[X_1 + \cdots + X_n] -
        \gamma] \leq e^{-\eps \gamma},
      \end{align*}
      and
      \begin{align*}
        \Pr[X_1 + \cdots + X_n \geq (1+\eps)\E[X_1 + \cdots + X_n] +
        \gamma] \leq e^{-\eps \gamma}.
      \end{align*}
    \end{quote}
  }
    ,
    \[\Pr[d_2(S)\le (1-\eps)d_1(S) - \gamma] \le e^{-\eps \gamma / \mu} = n^{-\gamma k / (c_\mu\eps\tau)}\]
    \[\Pr[d_2(S)\ge (1+\eps)d_1(S) + \gamma] \le e^{-\eps \gamma / \mu} = n^{-\gamma k / (c_\mu\eps\tau)}\]
    Say that set $S$ is $(\epsilon,\gamma)$-approximated if $(1-\epsilon)d_1(S)-\gamma\le d_2(S)\le(1+\epsilon)d_1(S)+\gamma$. For $\gamma=\eps\tau|S|/2k$, the probability that a set $S$ is not $(\eps, \gamma)$-approximated is at most $2n^{-c_0|S|}$, where $c_0=\frac{1}{2c_\mu}$. Note that there are $O(n^{|S|})$ sets of size $|S|$. For sufficiently large $c_0$ (sufficiently small $c_\mu$), we can apply union bound to conclude that $(\eps, \gamma)$-approximation holds for all subsets $S$, i.e., w.h.p.\ for all $S\subseteq V$,
        \[|d_2(S)-d_1(S)|\le \eps d_1(S) + \frac{\eps\tau|S|}{2k}\]
    Observe that $d_3(S) = d_2(S) + |S|\cdot \frac{\eps\tau}{2k} $ for all $t$-cuts $S$ by construction, so this implies
    \[(1-\eps)d_1(S) \le d_3(S) \le (1+\eps)d_1(S) + \frac{\eps\tau|S|}{k}\]
    
    Next, we prove properties assuming the above high probability event holds.    
    We assumed there exists a $t$-cut $S$ with $d_1(S)=d^+(S)<\tau$ and $|S|\le k$. Then,
    \[d_3(S) \le (1+\eps)d_1(S)+\frac{\eps\tau|S|}{k} \le (1+2\eps)\tau\]
    After scaling down by $\mu$, the min cut value is $\le d_3(S)/\mu = O(\tau/\mu)=O(k\log n)$. This establishes property~1.

    For property 2, we divide into two cases. We discuss the cut value in $d_3$, which is equivalent to $G_0$ up to scaling by $\mu$. Let $\lambda$ be the $t$-mincut value.

    Case 1: $\lambda \ge (1-\eps)\tau$. Then, $S$ is an $(1+O(\eps))$-approximate $t$-mincut.

    Case 2: $\lambda < (1-\eps)\tau$. Suppose $S'$ is the $t$-mincut in $G_0$ with $d_3(S')=\lambda$. 
    We have $d_1(S') \le \frac{1}{1-\eps}d_3(S') < \tau$. So $S'$ satisfies property 2.
\end{proof}

\begin{lemma}\label{lem:tree-packing}
    Given a digraph $G$ with integer edge capacities, a fixed root vertex $t$ and a number $k$ such that the (unknown) $t$-mincut value is $\lambda \le k$, we can construct in $\tO(mk)$ time a fractional packing of $t$-arborescences of value $\ge (1-\eps)\lambda$.
\end{lemma}
\begin{proof}
    We can write the fractional $t$-arborescences packing problem as a packing LP. Denote $\mathcal{A}=\{A_1,\ldots, A_N\}$ as the family of all $t$-arborescences. We represent a fractional packing of $t$-arborescence as a vector $x\in \mathbb{R}^N$ where $x_i$ is the coefficient of $A_i$. Define the value of an arborescence packing as the sum of all coefficients $\val(x)=\sum_{i=1}^N x_i$.
    The goal is to maximize $\val(x)$, under the constraints that each edge $j$ can be fractionally used up to its capacity $w(j)$.

    Equivalently, we can scale down all coefficients by $\val(x)$ for a solution $x$. Then, the goal becomes to minimize among all edges $j$ the ratio between usage of $j$ and $w(j)$.
    Our problem can be stated in the framework of a standard packing problem:
    \begin{defn}[Packing problem \cite{Young95}]
        For convex set $P\subseteq \mathbb{R}^n$ and nonnegative linear function $f:P\to\mathbb{R}^m$, the packing problem aims to find $\gamma^* = \min_{x\in P}\max_{j\in[m]}f_j(x)$ i.e., the solution in $P$ that minimizes the maximum value of $f_j(x)$ over all $j$. The width of the packing problem $(P, f)$ is defined as $\omega = \max_{j\in[m],x\in P} f_j(x)- \min_{j\in[m],x\in P} f_j(x)$.
    \end{defn}
    In our setting, $P=\{x\in \mathbb{R}^N: \val(x)=1, x\ge 0\}$ is the convex hull of all $t$-arborescences, and $f_j(x)=\frac{\sum_{i\in [N]:j\in A_i}x_i}{w(j)}$ is the ratio of usage over capacity of an edge $j$. The width is at most $1/w_{min}$, where $w_{min}$ is the minimum edge capacity. Because the edge capacities are integers, the width is at most 1.

    Next we describe the packing algorithm \cite{Young95}. Maintain a vector $y\in \mathbb{R}^m$, initially set to $y=1$. In each iteration, find $x=\arg\min_{x\in P}\sum_j y_jf_j(x)$, and then add $x$ to set $S$ and replace $y$ by the vector $y'$ defined by $y'_j = y_j(1+\eps f_j(x)/\omega)$. After a number of iterations, return $\bar{x}\in P$, the average of all the vectors $x$ over the course of the algorithm. The lemma below upper bounds the number of iterations that suffice:
    \begin{lemma}[Corollary 6.3 of \cite{Young95}] \label{lem:packing}
        After $\lceil \frac{(1+\eps)\omega\ln m}{\gamma^*((1+\eps)\ln(1+\eps)-\eps)}\rceil$ iterations of the packing algorithm, $\bar{\gamma} = \max_j f_j(\bar{x}) \le (1+\eps)\gamma^*$.
    \end{lemma}

    By duality between $t$-arborescence packing and minimum $t$-cut, we have the value of maximum $t$-arborescence packing is $\lambda$ (Corollary 2.1 of \cite{Gabow95}), and the optimal value of scaled packing problem is $\gamma^*=\frac{1}{\lambda}$. We run $C\cdot k\log n\ge C\cdot \lambda \log n$ iterations of the packing algorithm for a sufficiently large constant $C$. Then by \Cref{lem:packing}, the value of output fractional packing is $\le (1+\eps)\gamma^*$. After scaling back the solution to a fractional packing, its value is $\ge \frac{1}{1+\eps}\lambda$.

    Next we analyze the time complexity. Each iteration computes $x=\arg\min_{x\in P}\sum_j y_j f_j(x)$, which can be solved by a minimum cost $t$-arborescence algorithm in $O(m+n\log n)$ time \cite{GGST86}. We run $\tO(k)$ iterations, so the total running time is $\tO(mk)$.
\end{proof}

\begin{lemma}\label{lem:1resp-tree-packing}
    Given a capacitated digraph $G$ with a fixed root vertex $t$ and a number $\tau > 0$, suppose there exists a $t$-cut $S$ with $d^+(S)< \tau$ and $|S|\le k$.
    In $\tO(mk)$ time, we can find $O(\log n)$ $t$-arborescences on vertex set $V_0\supseteq S^*$, such that the following holds w.h.p.: There exists a $t$-cut $S''$ with $d^+(S'')< \tau$ and a $t$-arborescence $T$ in the packing, such that $T$ 1-respects $S''$.
\end{lemma}
\begin{proof}
    We first apply \Cref{lem:partial-sparsify} to obtain sparsifier $G_0$ and guarantee some $S'$ with $d^+_G(S')< \tau$ has $d^+_{G_0}(S') \le (1+\eps)\lambda$, where $\lambda$ is the $t$-mincut value in $G_0$, and $\eps$ is set to be 0.1.
    Then, we apply \Cref{lem:tree-packing} to obtain a fractional arborescence packing $\mathcal{T}$ with value $\ge (1-\eps)\lambda$.
    The expected intersection size $\E_{T\sim \mathcal{T}}[d^+_{T}(S')] \le \frac{(1+\eps)\lambda}{(1-\eps)\lambda}\le 1.5$. By Markov's inequality, with constant probability a random tree $T\in \mathcal{T}$ has intersection size $< 2$, which must be 1 as an integer. So, we have constant probability to sample a 1-respecting $t$-arborescence, and we can obtain a 1-respecting $t$-arborescence w.h.p.\ by sampling $O(\log n)$ arborescences from the packing.
    Finally, let $S''$ be the minimum $t$-cut 1-respected by $T$, then $d^+_G(S'')\le d^+_G(S')< \tau$.
\end{proof}

By applying \Cref{lem:1-resp-mincut-alg} on each of the $O(\log n)$ $t$-arborescences sampled and output the minimum $t$-cut found, we obtain a $t$-cut $S'$ with $d^+(S')<\tau$.

Next, we analyze the running time.
The partial sparsification step in \Cref{lem:partial-sparsify} takes $\tO(m)$ time. The tree packing step in \Cref{lem:tree-packing} takes $\tO(mk)$ time. Finally, computing 1-respecting min cut on $O(\log n)$ trees takes polylogarithmic calls to max flow by \Cref{lem:1-resp-mincut-alg}. In conclusion, the running time is $\tO(mk)$ plus polylogarithmic calls to max flow.

We have established \Cref{lem:alg-sparsify}.

%% file: trubin.tex
Trubin's min-ratio cut algorithm has aesthetic similarities to ours. Our algorithm proceeds by finding a dense core, contracting, and repeating this process until we build the whole cut hierarchy. Trubin's algorithm iteratively finds a so-called ``maximal strength HS-induced subgraph,'' contracts that subgraph, and repeats to build a cut hierarchy. Here a maximal strength HS-induced subgraph is an induced subgraph of $G$ with maximal strength and with strength equal to its fractional arboricity. (The fractional arboricity of a graph is the density of its skew-densest subgraph.) Trubin proves that the top level  of the hierarchy exhibits the same cost min-ratio cut as that of the original graph via the following lemma.
\begin{lemma}[\cite{Trubin93}] \label{lem: trubin}
Let $S \subseteq V$. The strength of $G$ is equal to the strength of $G/S$ if the strength of $G[S]$ is at least the strength of $G$.
\end{lemma}

The maximum strength of an HS-induced subgraph is in fact the fractional arboricity of $G$. It is easy to show this using our results. In particular, let $S \subseteq V$ be of maximum size such that 
\[
\rho(S) = \max_{U \subseteq V \text{ nonempty}} \rho(U).
\]
Then, trivially $S$ is a maximum skew-densest subset of $G$. As such, by~\Cref{fact:core-exist} and~\Cref{lem:core-is-star}, we get that $S$ is a star set of the canonical cut hierarchy. Consequently, $\rho(S)$ is equal to the strength of $G[S]$---its min-ratio cut is the trivial multi-way cut into singletons as it is a star set. By our choice of $S$, the strength of $G[S]$ is then equal to its fractional arboricity. 

Finally, since the fractional arboricity of $G$ is at least $\rho(V)$ and in turn at least its strength (by considering the trivial multi-way cut), we obtain correctness of the non-optimized version of Trubin's algorithm.

However, the more optimized algorithm given in~\cite{Trubin93} is erroneous. We sketch the algorithm and then give a counterexample disproving its correctness. First, for each $v \in V$ and $\lambda \in [0,\infty)$ solve 
\begin{equation}
\max_{v \in S}\{ c(E[S]) - \lambda(|S| - 1)\}.  \label{eq:trubin}
\end{equation}
Solutions for a given $v$ can be represented by an increasing sequence of sets 
\[
\{v\} = X(\lambda_1^v, v) \subseteq X(\lambda_2^v, v) \cdots X(\lambda_{\ell_v}^v, v) = V,
\]
where $\lambda_1^v \geq \lambda_2^v \geq \cdots \geq \lambda_{\ell_v}^v = 0$ and, for $i \geq 2$, $X(\lambda_i^v, v)$ maximizes~\cref{eq:trubin} for $\lambda \in [\lambda_i^v, \lambda_{i-1}^v]$.  We also have $X(\lambda_1^v, v)$ maximizes~\cref{eq:trubin} for $\lambda \geq \lambda_1$.

Next, adhere to the following algorithm:
\begin{enumerate}
    \item Let $u \in V$ with $\lambda_1^u$ maximum over vertices with corresponding sequence of sets of length at least $2$ (if none exist, terminate).
    \item Contract $X(\lambda_2^u, u)$ in $G$ (setting $G$ to $G/X(\lambda_2^u, u)$), and replace all instances of elements in $X(\lambda_2^u, u)$ in computed sets with $u$. For $u$'s increasing sequence of sets, remove the first set and reduce the indices of the $\lambda_i^u$'s by $1$. Return to step 1.
\end{enumerate}
At the end, the final contracted graph should have each vertex a side of a min-ratio cut of $G$.

To show that this algorithm is incorrect, we consider a small graph, construct the sequences of sets corresponding to each vertex by hand, and then show that the algorithm does not terminate in a min-ratio cut of the original graph. Consider a path on vertices $a, b, c, d$ with edges of weight $2$, $1$, and $100$ between $a,b$; $b,c$; and $c,d$, respectively. We can compute the increasing sequences of sets by hand:
\begin{itemize}
    \item For $a$, we get $\{a\} \subseteq \{a,c,d\} \subseteq \{a,b,c,d\}$ with $\lambda_1^a = 50, \lambda_2^a = 3,$ and $\lambda_3^a = 0$.
    \item For $b$, we get $\{b\} \subseteq \{b,c,d\} \subseteq \{a,b,c,d\}$, with $\lambda_1^b = 50.5, \lambda_2^b = 2, $ and $\lambda_3^b = 0$.
    \item For $c$, we get $\{c\} \subseteq \{c,d\} \subseteq \{a,b,c,d\}$ with $\lambda_1^c = 100, \lambda_2^c = 1.5, $ and $\lambda_3^c = 0$.
    \item Finally, for $d$, we get $\{d\} \subseteq \{c,d\} \subseteq \{a,b,c,d\}$ with $\lambda_1^d = 100, \lambda_2^d = 1.5, $ and $\lambda_3^d = 0$.
\end{itemize}
Observe that the min-ratio cut of this graph has sides $\{a,b\}$ and $\{c,d\}$. However, we can see that it is not possible to follow the algorithm and contract $\{a,b\}$ into a single node: $\{a,b\}$ never appears among the sequences over the course of the algorithm.

%% file: main.bbl
\newcommand{\etalchar}[1]{$^{#1}$}
\begin{thebibliography}{VDBCP{\etalchar{+}}23}

\bibitem[Bar92]{Barahona92}
Francisco Barahona.
\newblock Separating from the dominant of the spanning tree polytope.
\newblock {\em Operations Research Letters}, 12(4):201--203, 1992.

\bibitem[BF20]{BlumenstockF20}
Markus Blumenstock and Frank Fischer.
\newblock A constructive arboricity approximation scheme.
\newblock In {\em International Conference on Current Trends in Theory and
  Practice of Informatics}, pages 51--63. Springer, 2020.

\bibitem[CKL{\etalchar{+}}24]{chen2023almost}
Li~Chen, Rasmus Kyng, Yang~P Liu, Simon Meierhans, and Maximilian~Probst
  Gutenberg.
\newblock Almost-linear time algorithms for incremental graphs: Cycle
  detection, sccs, $ s $-$ t $ shortest path, and minimum-cost flow.
\newblock {\em STOC'24}, 2024.

\bibitem[CLN{\etalchar{+}}22]{CLNPQS21}
Ruoxu Cen, Jason Li, Danupon Nanongkai, Debmalya Panigrahi, Thatchaphol
  Saranurak, and Kent Quanrud.
\newblock Minimum cuts in directed graphs via partial sparsification.
\newblock In {\em 2021 IEEE 62nd Annual Symposium on Foundations of Computer
  Science (FOCS)}, pages 1147--1158, 2022.

\bibitem[dVC25]{deVos25}
Tijn de~Vos and Aleksander~BG Christiansen.
\newblock Tree-packing revisited: Faster fully dynamic min-cut and arboricity.
\newblock In {\em Proceedings of the 2025 Annual ACM-SIAM Symposium on Discrete
  Algorithms (SODA)}, pages 700--749. SIAM, 2025.

\bibitem[Epp94]{Eppstein94}
David Eppstein.
\newblock Arboricity and bipartite subgraph listing algorithms.
\newblock {\em Information processing letters}, 51(4):207--211, 1994.

\bibitem[Gab95]{Gabow95}
Harold~N. Gabow.
\newblock Algorithms for graphic polymatroids and parametric s-sets.
\newblock In {\em Proceedings of the Sixth Annual ACM-SIAM Symposium on
  Discrete Algorithms}, SODA '95, page 88–97, USA, 1995.

\bibitem[GGST86]{GGST86}
Harold~N Gabow, Zvi Galil, Thomas Spencer, and Robert~E Tarjan.
\newblock Efficient algorithms for finding minimum spanning trees in undirected
  and directed graphs.
\newblock {\em Combinatorica}, 6(2):109--122, 1986.

\bibitem[GM95]{GabowManu95}
Harold~N. Gabow and K.~S. Manu.
\newblock Packing algorithms for arborescences (and spanning trees) in
  capacitated graphs.
\newblock In {\em Integer Programming and Combinatorial Optimization}, pages
  388--402. Springer Berlin Heidelberg, 1995.

\bibitem[Gol84]{Goldberg84}
A.~V. Goldberg.
\newblock Finding a maximum density subgraph.
\newblock Technical report, University of California at Berkeley, USA, 1984.

\bibitem[GW92]{DBLP:journals/algorithmica/GabowW92}
Harold~N. Gabow and Herbert~H. Westermann.
\newblock Forests, frames, and games: Algorithms for matroid sums and
  applications.
\newblock {\em Algorithmica}, 7(5{\&}6):465--497, 1992.

\bibitem[NW61]{Nash-Williams61}
C.~St.J.~A. Nash-Williams.
\newblock {Edge-Disjoint Spanning Trees of Finite Graphs}.
\newblock {\em Journal of the London Mathematical Society}, s1-36(1):445--450,
  01 1961.

\bibitem[NW64]{nash-williams}
C.~St.J.~A. Nash-Williams.
\newblock Decomposition of finite graphs into forests.
\newblock {\em Journal of the London Mathematical Society}, s1-39(1):12--12,
  1964.

\bibitem[PQ82]{DBLP:journals/networks/PicardQ82}
Jean{-}Claude Picard and Maurice Queyranne.
\newblock A network flow solution to some nonlinear 0-1 programming problems,
  with applications to graph theory.
\newblock {\em Networks}, 12(2):141--159, 1982.

\bibitem[PST95]{PlotkinST95}
Serge~A Plotkin, David~B Shmoys, and {\'E}va Tardos.
\newblock Fast approximation algorithms for fractional packing and covering
  problems.
\newblock {\em Mathematics of Operations Research}, 20(2):257--301, 1995.

\bibitem[PW84]{DBLP:journals/mp/PadbergW84}
Manfred Padberg and Laurence~A. Wolsey.
\newblock Fractional covers for forests and matchings.
\newblock {\em Math. Program.}, 29(1):1--14, 1984.

\bibitem[Sch03]{Schrijver03}
A.~Schrijver.
\newblock {\em Combinatorial Optimization}.
\newblock Springer, 2003.

\bibitem[Tho01]{Thorup01}
Mikkel Thorup.
\newblock Fully-dynamic min-cut.
\newblock In {\em Proceedings of the Thirty-Third Annual ACM Symposium on
  Theory of Computing}, STOC '01, page 224–230, 2001.

\bibitem[Tho08]{Thorup08}
Mikkel Thorup.
\newblock Minimum k-way cuts via deterministic greedy tree packing.
\newblock In {\em Proceedings of the Fortieth Annual ACM Symposium on Theory of
  Computing}, STOC '08, page 159–166, 2008.

\bibitem[Tru93]{Trubin93}
VA~Trubin.
\newblock Strength of a graph and packing of trees and branchings.
\newblock {\em Cybernetics and Systems Analysis}, 29(3):379--384, 1993.

\bibitem[Tut61]{Tutte61}
W.~T. Tutte.
\newblock {On the Problem of Decomposing a Graph into n Connected Factors}.
\newblock {\em Journal of the London Mathematical Society}, s1-36(1):221--230,
  01 1961.

\bibitem[VDBCP{\etalchar{+}}23]{van2023deterministic}
Jan Van Den~Brand, Li~Chen, Richard Peng, Rasmus Kyng, Yang~P Liu,
  Maximilian~Probst Gutenberg, Sushant Sachdeva, and Aaron Sidford.
\newblock A deterministic almost-linear time algorithm for minimum-cost flow.
\newblock In {\em 2023 IEEE 64th Annual Symposium on Foundations of Computer
  Science (FOCS)}, pages 503--514. IEEE, 2023.

\bibitem[WG16]{WorouG16}
Bio Mikaila~Toko Worou and J{\'e}r{\^o}me Galtier.
\newblock Fast approximation for computing the fractional arboricity and
  extraction of communities of a graph.
\newblock {\em Discrete Applied Mathematics}, 213:179--195, 2016.

\bibitem[You95]{Young95}
Neal~E Young.
\newblock Randomized rounding without solving the linear program.
\newblock In {\em Proceedings of the sixth annual ACM-SIAM symposium on
  Discrete algorithms}, pages 170--178, 1995.

\end{thebibliography}
